\documentclass[runningheads]{llncs}
\usepackage{tikz}
\usetikzlibrary{positioning,chains,fit,shapes,calc}
\usepackage{booktabs} % For formal tables
\usepackage[ruled]{algorithm2e} % For algorithms

\SetAlFnt{\small}
\SetAlCapFnt{\small}
\SetAlCapNameFnt{\small}
\SetAlCapHSkip{0pt}
\IncMargin{-\parindent}

\usepackage{macros_pan}

\let\llncssubparagraph\subparagraph
\let\subparagraph\paragraph
\usepackage[compact]{titlesec}
\let\subparagraph\llncssubparagraph

\newcommand{\bw}{\mathbf{w}}

\newcommand{\off}{\ensuremath{\operatorname{OFF-CCM}}\xspace}
\newcommand{\on}{\ensuremath{\operatorname{ON-CCM}}\xspace}
\newcommand{\ons}{\ensuremath{\operatorname{ON-CSM}}\xspace}

\newcommand{\algoff}{\ensuremath{\mathsf{ALG1}}\xspace}
\newcommand{\algon}{\ensuremath{\mathsf{ALG2}}\xspace}

\newcommand{\nna}{\ensuremath{\mathsf{ALG3}}\xspace}
\newcommand{\nnab}{\ensuremath{\mathsf{ALG3B}}\xspace}

\newcommand{\sG}{\ensuremath{\mathsf{G}}\xspace}

\newcommand{\proa}{\textbf{P1}\xspace}
\newcommand{\prob}{\textbf{P2}\xspace}
\newcommand{\proc}{\textbf{P3}\xspace}

\usepackage{url}            % simple URL typesetting
\usepackage{amsfonts}       % blackboard math symbols
\usepackage{nicefrac}       % compact symbols for 1/2, etc.
\usepackage{microtype}      % microtypography
\usepackage{xcolor}         % colors
\usepackage{subcaption}
\usepackage{multirow}
\usepackage{algorithmic}

\usepackage[sectionbib,numbers,compress]{natbib}
\newcommand{\bsw}{\ensuremath{\operatorname{ON-CSM}}\xspace}
\newcommand{\bbsp}{\ensuremath{\operatorname{BBM-2}}\xspace}

\begin{document}

\title{Two-Sided Capacitated Submodular Maximization in Gig Platforms\thanks{A preliminary version of this paper was accepted to the 19th {Conference on Web and Internet Economics} (WINE), 2023. Pan Xu was partially supported by NSF CRII Award IIS-1948157. The author would like to thank the anonymous reviewers for their valuable comments.}}

\author{Pan Xu}
\institute{Department of Computer Science, New Jersey Institute of Technology
\\
\email{pxu@njit.edu}}

\maketitle

\begin{abstract}
In this paper, we propose three generic models of capacitated coverage and, more generally, submodular maximization to study task-worker assignment problems that arise in a wide range of gig economy platforms. Our models incorporate the following features: (1) Each task and worker can have an arbitrary matching capacity, which captures the limited number of copies or finite budget for the task and the working capacity of the worker; (2) Each task is associated with a coverage or, more generally, a monotone submodular utility function. Our objective is to design an allocation policy that maximizes the sum of all tasks' utilities, subject to capacity constraints on tasks and workers. We consider two settings: offline, where all tasks and workers are static, and online, where tasks are static while workers arrive dynamically. We present three LP-based rounding algorithms that achieve optimal approximation ratios of $1-1/\mathsf{e} \sim 0.632$ for offline coverage maximization, competitive ratios of $(19-67/\mathsf{e}^3)/27\sim 0.580$ and $0.436$ for online coverage and online monotone submodular maximization, respectively. 
\end{abstract} 
%Extensive experiments on both real and synthetic datasets demonstrate the effectiveness of our LP-based algorithms and validate our theoretical predictions.

%Crowdsourcing markets such as Amazon Mechanical Turk, 99designs and Cad Crowd have gained great popularity with the help of the Internet. A basic setting is that we have a set of tasks and workers in the system, and completion of each task will contribute some revenue to the system. Thus, a natural goal is to design an allocation policy such that
\section{Introduction}
The gig economy has flourished due to the popularity of smartphones over the last few decades. Examples range from ride-hailing services such as Uber and Lyft to food delivery platforms like Postmates and Instacart, and freelancing marketplaces such as Upwork. These platforms have leveraged the widespread use of smartphones to efficiently connect service providers and customers, creating a significant number of convenient and flexible job opportunities. At its core, the gig economy involves a system with service requesters (tasks) and service providers (workers), where completing each task contributes revenue to the system. Therefore, a fundamental issue in revenue management is designing a matching policy between tasks and workers that maximizes the total revenue from all completed tasks. A common strategy is assigning a weight to each task, representing the profit obtained from completing it, and formulating the objective as the maximization of a linear function that represents the total revenue over all completed tasks~\cite{ho2012online,assadi2015online,xu2017budgeted}. However, this paper proposes a different approach by associating each task with a specific monotone submodular function over the set of workers, thus updating the objective to the maximization of the sum of utility functions over all tasks. Let's consider the following motivating examples.

\noindent \textbf{Assigning diverse reviewers to academic papers}.~\citet{ahmed2017diverse} have studied how to assign a set of relevant and diverse experts to review academic papers/proposals.  As mentioned there, experts' diversity plays a key role in maintaining fairness in the final decision (acceptance or rejection of the paper). A general approach there is to select a ground set of $K$ features describing reviewers (\eg  affiliations, research focuses, and demographics). For each reviewer $j$, we label it with a binary vector $\chi_j \in \{0,1\}^K$ such that $\chi_{j,k}=1$ iff reviewer $j$ covers feature $k$.  For each paper-feature pair $(i,k)$, we associate it with a non-negative weight $w_{i,k}$ reflecting the degree of importance of feature $k$ to paper $i$. Under this configuration, our goal of forming a diverse team of experts for each paper can be formulated as $\max \sum_i g_i(S_i)$, where $S_i$ is the set of reviewers assigned to paper $i$, and $g_i (S_i)=\sum_k w_{i,k} \min \big(1, \sum_{j \in S_i} \chi_{j,k}\big)$ \bluee{denotes} the total weight of all covered features for paper $i$. 

\noindent \textbf{Multi-skilled task-worker assignments}. Consider a special class of crowdsourcing markets featuring that every task and worker is associated with a set of skills~\cite{barnabo2019algorithms,cheng2016task,anagnostopoulos2012online}. A natural goal there is to assign each task \bluee{to} a set of workers such that the task has as many skills covered as possible. A typical approach is as follows. We identify a ground set of $K$ skills; each task  $i$ and each worker  $j$ are labeled with a binary vector $\chi_i, \chi_j \in \{0,1\}^K$, where $\chi_{i,k}=1$ and  $\chi_{j,k}=1$  indicate that skill $k$ is requested by task $i$ and possessed by worker $j$, respectively. Thus, our goal can be formulated as $\max \sum_i g_i(S_i)$, where $S_i$ is the set of workers assigned to task $i$, and $g_i(S_i)=\sum_{k} \chi_{i,k} \cdot \min\big(1, \sum_{j \in S_i} \chi_{j,k}\big)$ denoting the number of skills requested and covered for task $i$. More generally, suppose for each task-skill pair $(i,k)$, we associate it with a non-negative weight $w_{i,k}$ reflecting the importance of skill $k$ to task $i$. Our generalized objective then can be reformulated as $\max \sum_i g_i(S_i)$, where $g_i(S_i)=\sum_{k} w_{i,k} \min\big(1, \sum_{j \in S_i} \chi_{j,k}\big)$ denoting the total weight of all covered skills for task $i$.

\noindent \textbf{Diversity maximization among online workers}. In many real-world freelancing platforms, it is highly desirable to assign a set of diverse workers to each task (\eg labeling an image and soliciting public views for a particular political topic). This becomes particularly prominent in the context of healthcare when we need to crowdsource a set of highly diversified online volunteers for medical trials. As reported in \cite{fair-scm}, biases in health data are common and can be life-threatening when the training data feeding machine-learning algorithms lack diversity.~\citet{DBLP:journals/corr/abs-2002-10697} considered an online diverse team formation problem, where the overall goal is to crowdsource a  team of diverse online workers for every (offline) task.  They defined a set of $K$ features (or clusters) reflecting workers' demographics  and thus, each work  $j$ can be modeled by a binary vector $\chi_j \in \{0,1\}^K$ such that $\chi_{j,k}=1$ indicates that work $j$ has feature $k$ (\ie belongs to cluster $k$). They associated each task-feature pair $(i,k)$ a non-negative weight $w_{i,k}$ reflecting the utility of adding one worker with feature $k$  to the team for task $i$. They proposed a utility function on each task $i$ as $g_i(S_i)=\sum_k \sqrt{w_{i,k}\cdot \sum_{j \in S_i} \chi_{j,k}}$, where $S_i$ is the team of workers assigned to task $i$. The overall objective in~\cite{DBLP:journals/corr/abs-2002-10697} is then formulated as $\sum_i g_i(S_i)$.

Here are some similarities and dissimilarities among the three examples above. First, all objectives can be formulated as the maximization of an unweighted/weighted coverage function or a more general monotone submodular function. Second, each agent in the system is associated with a finite capacity. In the context of the paper-reviewer assignment, each reviewer $j$ has a matching capacity $b_j \in \mathbb{Z}{+}$ that captures the maximum number of papers reviewer $j$ can handle, while each paper $i$ also has a capacity $b_i \in \mathbb{Z}{+}$ reflecting the number of reviewers we should allocate for paper $i$ due to the shortage of available reviewers. Similar capacity constraints exist in a wide range of real-world crowdsourcing markets: each task $i$ and worker $j$ practically have a matching capacity $b_i$ and $b_j$, respectively, where $b_i$ models the limited budget or copies of task $i$, and $b_j$ reflects the working capacity of worker $j$. Third, the first application differs from the other two in the arrival setting of agents. Consider the paper-reviewer assignment for a big conference, for example. Generally, all information about tasks (i.e., the papers to review) and workers (i.e., available reviewers) is accessible before any matching decisions, and thus, the setting is called static or offline. In contrast, in most real-world crowdsourcing markets, only part of the agents are static such as tasks, whose information is known in advance, while some agents like workers join the system dynamically~\cite{ho2012online, xu2017budgeted}. Inspired by all the insights above, we propose three generic models as follows. \emph{Throughout this paper, we state our models in the language of matching tasks and workers in a typical crowdsourcing market}.

\noindent \textbf{OFFline Capacitated Coverage Maximization} (\off).
Suppose we have a bipartite graph $G=(I,J,E)$, where $I$ and $J$ denote the sets of tasks and workers, respectively.  An edge $e=(i,j)$ indicates the feasibility or interest for worker $j$ to complete task $i$. We have a ground set $\cK$ of $K$ features. Each worker $j$ is captured by a binary vector $\chi_j \in \{0,1\}^K$ such that $\chi_{jk}=1$ iff it covers feature $k$. Each task $i$ has a weight vector $\bw_i=(w_{ik})$, where $w_{ik} \in [0,1]$ reflects the importance/weight of feature $k$ with respect to $i$. Each worker $j$ (and task $i$) has an integer capacity $b_j$ ($b_i$), which means that it can be matched with at most $b_j$ ($b_i$) different tasks (workers). Here $b_j$ reflects the working capacity of worker $j$, while $b_i$ models the number of copies or budget of task $i$. Consider an allocation $\x=(x_{ij}) \in \{0,1\}^{|E|}$, where $x_{ij}=1$ with $(i,j) \in E$ means $j$ is assigned to $i$. For each $i \in I$, let $\cN_i=\{j \in J, (i,j) \in E\}$ be the set of neighbors of $i$; similarly for $\cN_j$ with $j \in J$. We say $\x$  is \emph{feasible} or \emph{valid} iff $\sum_{i' \in \cN_j} x_{i',j} \le b_j$ and
 $\sum_{j' \in \cN_i} x_{i,j'} \le b_i$ for all $i \in I$ and $j \in J$.  We define the utility of task $i$ under $\x$ as $g_i(\x)=\sum_{k \in \cK} w_{ik} \cdot \min(1, \sum_{j: \chi_{jk}=1} x_{ij})$, \ie the total sum of weights of features covered under $\x$, and the total resulting utility under $\x$ as  $g(\x)=\sum_{i \in I} g_i(\x)$, \ie  the sum of utilities over all tasks.  Note that under the offline setting, an input instance can be specified as $\cI=(G, \{\chi_j\}, \{\bw_i\}, \{b_i, b_j\})$, which is fully accessible. We aim to compute a feasible allocation $\x$ such that $g(\x)$ is maximized.

\noindent \textbf{ONline Capacitated Coverage Maximization} (\on).
The basic setting here is the same as \off. Specifically, we assume all information of $(G, \{\chi_j\}, \{\bw_i\}, \{b_i, b_j\})$ is fully known to the algorithm, but that is only part of the input. The graph $G=(I,J,E)$ in our case should be viewed as a compatible graph, where $I$ and $J$ denote the sets of \emph{types} of \emph{offline} tasks and \emph{online} workers, respectively. Tasks are static, while workers arrive dynamically following a \emph{known independent identical distribution} (KIID) as specified as follows. We have a finite time horizon $T$, and during each round $t \in \{1,2,\ldots, T\}$ one single worker (of type) $\hat{j}$ will be sampled (called $\hat{j}$ arrives) with replacement such that $\Pr[\hat{j}=j]=r_j/T$ for all $j \in J$ with  $\sum_{j \in J} r_{j}/T=1$. Here $r_j$ is called the \emph{arrival rate} of $j$. Note that the arrival distribution $\{r_j/T\}$ is assumed \emph{independent} and \emph{invariant} throughout the $T$ rounds, and it is accessible to the algorithm. The KIID arrival setting is mainly inspired \bluee{by} the fact that we can often learn the arrival distribution from historical logs~\cite{Yao2018deep,DBLP:conf/kdd/LiFWSYL18}. Upon the arrival of every online worker $j$, we have to make an \emph{immediate and irrevocable} decision: either reject $j$ or assign it to at most $b_j$ neighbors from $\cN_j$  (subject to capacity constraints from tasks as well). Our goal is to design an allocation policy $\ALG$ such that $\E[g(\X)]$ is maximized, where $\X$ is the allocation (possibly random) output by $\ALG$, and where the expectation is taken over the randomness in the online workers' arrivals and that in $\ALG$. 

\noindent \textbf{ONline  Capacitated Submodular Maximization} (\ons). The setting is almost the same as \on except that each task $i \in I$ is associated with a general \emph{monotone submodular} utility function $g_i$ over the \bluee{ground set} of $\cN_i$ (the set of neighbors of $i$). WLOG assume $g_i(\emptyset)=0$ for all $i \in I$.

Apart from crowdsourcing markets, capacitated coverage maximization and general submodular maximization have applications in promoting diversity in other domains. For instance, they are used in crowdsourcing test platforms, where a diverse set of users is needed to test mobile apps~\cite{xie2017cqm}. These models also find applications in online recommendations~\cite{puthiya2016coverage, ge2010beyond} and document clustering~\cite{abbassi2013diversity}, where diversity is desired. 
%\bluee{perhaps we can add a few recent papers here?!}

\section{Preliminaries, Main Contributions, and Other Related Work}

%\subsection{Preliminaries}
Throughout this paper, we assume the total number of online arrivals, denoted by $T$, is significantly large ($T \gg 1$), and \emph{part of our results are obtained after taking $T \rightarrow \infty$}. This assumption is commonly made and practiced in the study of online-matching models with maximization of linear objectives under KIID~\cite{brubach2020online, bib:Jaillet, bib:Manshadi, bib:Haeupler, bib:Feldman}.

\noindent\textbf{Approximation ratio}. For NP-hard combinatorial optimization problems, a powerful approach is \emph{approximation algorithms}, where the goal is to design an efficient algorithm (with polynomial running time) that guarantees a certain level of performance compared to the optimal solution. In the case of a maximization problem like \off, we denote an approximation algorithm and its performance as $\ALG$, and an optimal algorithm with no running-time constraint and its performance as $\OPT$. We say that $\ALG$ achieves an approximation ratio of at least $\rho \in [0,1]$ if $\ALG \ge \rho \cdot \OPT$ for any input instances.

\noindent \textbf{Competitive ratio} (CR). Competitive ratio is a commonly used metric to evaluate the performance of online algorithms. For a given algorithm \ALG and an (online) maximization problem like \on as studied here, we denote the expected performance of \ALG on an instance $\cI$ as $\ALG(\cI)$, where the expectation is taken over the randomness in the arrivals of workers and that in \ALG. Similarly, we denote the expected performance of a \emph{clairvoyant optimal} as $\OPT(\cI)$. We say that \ALG achieves a CR of at least $\rho \in [0,1]$ if $\ALG(\cI) \ge \rho \cdot \OPT(\cI)$ for any input instance $\cI$. It is important to note that \ALG is subject to the real-time decision-making requirement, while a clairvoyant optimal \OPT is exempt from that constraint. Thus, CR captures the gap in expected performance between \ALG and \OPT due to the instant-decision requirement. In the example below, we provide an instance of \on and demonstrate that the natural algorithm \gre achieves a CR of zero. By definition, \gre will assign every arriving worker $j$ to at most its $b_j$ \emph{safe} neighbors (those that still have \bluee{the} capacity to admit workers) following a decreasing order of the marginal contribution to task utility functions resulting from adding $j$.

\begin{figure}[ht!]
\begin{minipage}{.5\linewidth}
 \begin{tikzpicture}

 \draw (0,0) node[minimum size=0.2mm,draw, thick, circle] {$1$};
  \draw (3,0) node[minimum size=0.2mm,draw, thick, circle] {{$1$}};

    \draw (3,-1) node[minimum size=1mm,draw, thick, circle] {$2$};
            \draw (3,-3) node[minimum size=1mm,draw, thick,  circle] {$n$};
             \node [blue,above] at (1.7,0) {$w=1$};
                 \node [blue,above] at (2,-0.7) {$\ep$};
                    \node [blue,above] at (2,-1.4) {$\ep$};
                       \node [blue,above] at (2,-2.1) {$\ep$};
               \node [above] at (0,0.3) {$I$};
               \node [above] at (3,0.3) {$J$};
     
\draw[ultra thick,-] (0.3,0)--(2.7,0);
\draw[ultra thick,-] (0.3,0)--(2.7,-1);
\draw[ultra thick, -] (0.3,0)--(2.7,-3);
\draw[ultra thick, -] (0.3,0)--(2.7,-2);
\draw [dotted, ultra thick](2.7,-2) -- (3.3,-2);
\end{tikzpicture}
\end{minipage}\hspace{-0.5in}
\begin{minipage}{.5\linewidth}
\begin{align*}
&I=\{1\}, J=\cK=[n]:=\{1,2,\ldots,n\}; \\
& \chi_j=\mathbf{e}_j ~(\mbox{the $j$th stardard unit vector}),\forall j \in [n]; \\
& w_{i=1,k=1}=1, w_{i=1,k}=\ep, \forall 2 \le k \le n;\\
&b_{i=1}=b_j=1, \forall j \in [n];\\
%&W_{(i,j_\ell)}=\ep, \mbox{ with probability $1$}, \forall 1<\ell \le n;\\
&T=n, r_j=1, \forall j \in [n];\\
&\gre\le \ep+1/n;\\
&\OPT \ge 1-1/\sfe.
\end{align*}
\end{minipage}
\caption{A toy example where  \gre achieves a competitive ratio of zero for $\on$.} \label{fig:exam-1}
\end{figure}
 
 \begin{example}\label{exam:1}
Consider an instance of \on shown in Figure~\ref{fig:exam-1}. We have a star graph $G=(I,J,E)$ with $|I|=1$, $|J|=|\cK|=n$, and $b_i=b_j=1$ for $i=1$ and all $j \in [n] \doteq \{1,2,\ldots,n\}$. Each worker (of type) $j$ covers a single feature $k=j$ for all $j \in [n]$. The weight $w_{k} \doteq w_{i=1,k}=1$ if $k=1$ and $w_k=\ep$ if $1<k \le n$. Set $T=n$ and $r_j=1$ for all $j \in [n]$, i.e., during each time $t \in [n]$, a worker $j$ will arrive uniformly at random. In our context, \gre can be interpreted as assigning whatever arriving worker to the only task in the system.

Observe that at time $t=1$, with respective probabilities of $1/n$ and $1-1/n$, it is the worker (of type) $j=1$ and $1<j \le n$ that will arrive. By the nature of \gre, we will end up with a total utility of $1$ and $\ep$ with respective probabilities of $1/n$ and $1-1/n$. Thus, we claim that $\E[\gre]=1/n \cdot 1+(1-1/n)\cdot \ep$. Recall that \bluee{the clairvoyant optimal \OPT has} the privilege to optimize its decision after observing the full arrival sequence of all workers. Note that with probability $1-1/\sfe$, worker $j=1$ will arrive at least once. Thus, \OPT will end up with a total utility of $1$ and $\ep$ with respective probabilities of $1-1/\sfe$ and $1/\sfe$. Therefore, we claim that the expected performance of \OPT should be $\E[\OPT]=1 \cdot(1-1/\sfe)+\ep \cdot 1/\sfe$. By definition, we conclude that \gre achieves a CR of zero when $\ep \rightarrow 0$ and $n \rightarrow \infty$ (i.e., it has arbitrarily bad performance).
 \end{example}
 
\noindent \textbf{Connections to existing models}. 
In the offline setting (\off), the closest model, to the best of our knowledge, is Submodular Welfare Maximization (SWM) introduced by~\citet{vondrak2008optimal}. The basic setting is as follows (rephrased in our language): We have a complete bipartite graph $G=(I,J)$, where each task $i$ is associated with a non-negative monotone submodular utility function $g_i$ over the ground set of $J$. Each task $i$ has an unbounded capacity ($b_i=\infty$), and each worker $j$ has a unit capacity ($b_j=1$). The goal is to compute a \emph{partition} ${S_i}$ of $J$ such that $\sum_{i \in I} g_i(S_i)$ is maximized. \citet{vondrak2008optimal} demonstrated that SWM can be represented as a special case of maximizing a general monotone submodular function subject to a matroid constraint, and they provided a randomized continuous greedy algorithm that achieves an optimal approximation ratio of $1-1/\sfe$. In comparison to SWM, our offline setting (\off) considers a special case where each $g_i$ is a weighted coverage function. Furthermore, we generalize SWM in the following three ways: (\textbf{F1}) $G$ can be any bipartite graph (not necessarily complete), (\textbf{F2}) the capacity $b_i$ of task $i$ can be any positive integer instead of $b_i=\infty$, and (\textbf{F3}) the capacity $b_j$ of worker $j$ can be any positive integer instead of $b_j=1$. Among these three new features, (\textbf{F2}) is perhaps the most non-trivial one since it essentially imposes a new partition-matroid constraint on all feasible allocations. \emph{It might be tempting to consider \off as a special case of SWM by creating $b_j$ copies for each worker $j \in J$ and introducing an uncapacitated version of the utility function for each task $i \in I$ as $\tilde{g}_i(S)=\max_{S': S' \subseteq S \cap \cN_i, |S'| \le b_i} g_i(S')$ for any $S \subseteq J$. However, as shown in Appendix~\ref{sec:app-a}, this reduction only yields a $(1-1/\sfe)^2$-approximate algorithm, which is significantly worse than what we present here, as shown in Theorem~\ref{thm:main-a}.}

Regarding the online setting, \citet{kapralov2013online} examined the online version of SWM under KIID, whether known or unknown. They demonstrated that \gre achieves an \emph{optimal} competitive ratio (CR) of $1-1/\sfe$. It is worth noting that the introduction of the two new features (\textbf{F2}) and (\textbf{F3}) to SWM each brings about significant algorithmic challenges to both \on and \ons scenarios. The example presented in Example~\ref{exam:1} suggests that after (\textbf{F2}) is introduced to online SWM, \gre experiences a substantial decrease in performance, shifting from being optimal (with a CR of $1-1/\sfe$) to being arbitrarily bad (with a CR of zero). Meanwhile, (\textbf{F3}) indicates that a worker can have a non-unit capacity, leading to the involvement of multiple assignments upon her arrival. This further complicates the design and analysis of algorithms.

\xhdr{Glossary of Notations}. A glossary of notations used throughout this paper is shown in Table~\ref{table:notation}.
\begin{table}[th!]
\center
\caption{A glossary of notations used throughout this paper.}
\label{table:notation}
\begin{tabular}{ll } 
 \hline
  $[n]$ & Set of integers $\{1,2,\ldots,n\}$ for a generic integer $n$.  \\[0.05cm]
$G=(I, J, E)$ & Input graph (accessible in advance), where $I$ ($J$) is the set of task (worker) types.   \\[0.1cm]
$\cK$ & Ground set of features with $|\cK|=K$. \\[0.05cm]
$\chi_j \in \{0,1\}^K$ & Characteristic  vector of worker (of type) $j$ with $\chi_{jk}=1$ indicating $j$ covers feature $k$.\\ [0.05cm]
$w_{ik} \in [0,1]$ & Weight of feature $k$ with respect to task (of type) $i$.\\ [0.05cm]
$b_{i}$ ($b_{j}$) & Matching capacity on task $i$ (worker $j$).\\ [0.05cm]
$\cN_i$  ($\cN_{j}$) & Neighbors of task $i$ (worker $j$) in the input graph with $\cN_i \subseteq J$ ($\cN_j \subseteq I$).\\ [0.05cm]
$E_i$  ($E_{j}$) & Set of edges incident to task $i$ (worker $j$) in the input graph.\\ [0.05cm]
$T$& Time horizon, \ie the number of online rounds with $T \gg 1$.\\ [0.05cm]
$r_j$ & Arrival rate of worker $j$ (in the online setting) such that $j$ arrives with probability $r_j/T$ every time. \\ [0.05cm]
$g_i$ & Generic monotone submodular utility function for task $i$ with $g_i(\emptyset)=0$.
\\ [0.05cm]
$\sfe$ vs.\ $e$ & The former represents the natural base with $\sfe \sim 2.718$, while the latter an edge $e \in E$.\\ [0.05cm]
%\off& OFFline Capacitated Coverage Maximization. \\ [0.05cm]
%\on& ONline Capacitated Coverage Maximization. \\ [0.05cm]
%\ons& ONline Capacitated Submodular Maximization. \\ [0.05cm]
 \hline
\end{tabular}
\end{table}

\noindent \textbf{Main Contributions}. 
In this paper, we propose three generic models of capacitated coverage and, more generally, submodular maximization, to study task-worker assignment problems that arise in gig economy platforms. For each of the three models, we construct a linear program (LP) that provides a valid upper bound for the corresponding optimal performance. Using these benchmark LPs as a foundation, we develop dependent-rounding-based (DR-based) sampling algorithms. Below are our main theoretical results. \emph{Throughout this paper, all fractional values are estimated with accuracy to the third decimal place}.

\begin{theorem}\label{thm:main-a}[Section~\ref{sec:off}]
There is a polynomial-time algorithm that achieves an \emph{optimal} approximation ratio of $1-1/\sfe$ for \off.
\end{theorem}

\textbf{Remarks on Theorem~\ref{thm:main-a}}. Note that \off can be viewed as a special case of maximization of a monotone submodular function subject to two \emph{partition} matroids.\footnote{The sum of weighted coverage utility functions over all tasks can be viewed as one single monotone submodular function over the ground set of all edges.} Meanwhile, \off strictly generalizes the classical Maximum Coverage Problem (MCP), which can be cast as a special case of \off with one single task. \citet{feige1998threshold} showed that MCP is NP-hard and cannot be approximated within a factor better than $1-1/\sfe$ unless $P=NP$. This suggests the optimality of the approximation ratio of $1-1/\sfe$ in Theorem~\ref{thm:main-a}. Observe that $1-1/\sfe$ is larger than the current best ratio for maximization of a monotone submodular function subject to two \emph{general} matroids, which is $1/2-\ep$ due to the work of \cite{lee10}. We believe our technique can be of independent interest and perhaps can be generalized to study the case of two or multiple general matroids.

\begin{theorem}\label{thm:main-1}[Section~\ref{sec:special}]
There is an algorithm  that achieves a competitive ratio of at least $(19-67 \cdot \sfe^{-3})/27 \sim 0.580$ for \on. In particular, it achieves an \emph{optimal}  competitive ratio (CR) of $1-1/\sfe$ for \on when every task has no capacity constraint.
\end{theorem}
\begin{theorem}\label{thm:main-2}[Section~\ref{sec:gen}]
There exists an algorithm that achieves a competitive ratio of at least $0.436$ for \ons when every task has constant capacity.  
\end{theorem}

\textbf{Remarks on Theorems~\ref{thm:main-1} and~\ref{thm:main-2}}. (1) \on with feature \textbf{(F2)} off (\ie no capacity constraint on tasks with $b_i=\infty$) still \emph{strictly} generalizes online SWM introduced in~\cite{kapralov2013online} due to features \textbf{(F1)} and \textbf{(F3)}. That being said, our algorithm achieves an optimal CR of $1-1/\sfe$ for \on with \textbf{(F2)} off, which matches the best possible CR for online SWM that was achieved by \greedy~\cite{kapralov2013online}. (2) Both \on and \ons strictly generalize the classical online (bipartite) matching under KIID when each task has a unit capacity (\ie all $b_i=1$). For online matching under KIID, \citet{bib:Manshadi} showed upper bounds on the competitive ratio over all possible adaptive and non-adaptive algorithms, which are $0.823$ and $1-1/\sfe\sim 0.632$, respectively. These two upper bounds apply to the $0.580$-competitive algorithm in Theorem~\ref{thm:main-1} and the $0.436$-competitive algorithm in Theorem~\ref{thm:main-2}.\footnote{Algorithms mentioned in Theorems~\ref{thm:main-1} and~\ref{thm:main-2} are both non-adaptive; see Section~\ref{sec:special} and~\ref{sec:gen} for more details.} (3) The constant-task-capacity assumption in Theorem~\ref{thm:main-2} is mainly inspired by real-world gig economy  platforms, where the capacity of every task is relatively small due to a finite number of copies and/or a limited budget~\cite{DBLP:journals/corr/abs-2002-10697,Faez-17}.

%%%
\begin{comment}
To complement \bluee{the} theoretical results above, we implement and compare our algorithms against several heuristics on both real and \bluee{synthetic} datasets. Experimental results demonstrate the effectiveness of our LP-based algorithms and confirm our theoretical predictions; see Appendix~\ref{app:exp} for more details. 
\end{comment}
%%%

\textbf{Main Techniques}.  Our algorithms for the offline and online settings both invoke the technique of dependent rounding (DR) as a subroutine, which was introduced by \citet{gandhi2006dependent}. Recall that $G=(I,J,E)$ is \bluee{bipartite}. Let $E_\ell$ \bluee{be} the set of edges incident to $\ell$ for $\ell \in I \cup J$. Suppose each edge $e \in E$ is associated with a fractional value $z_e \in [0,1]$.      DR is such a rounding technique that takes as input a fractional vector $\z=(z_{e}) \in [0,1]^{|E|}$ \bluee{and outputs} a random binary vector $\Z=(Z_{e}) \in \{0,1\}^{|E|}$ \bluee{which satisfies} the following properties. (\proa) \textbf{Marginal distribution}: $\E[Z_{e}]=z_{e}$ for all $e \in E$;
      (\prob) \textbf{Degree preservation}: $\Pr[Z_\ell \in \{\lfloor z_\ell \rfloor, \lceil z_\ell \rceil\}]=1$ for all $\ell \in I \cup J$, where $Z_\ell \doteq \sum_{e \in E_\ell} Z_e$ and $z_\ell \doteq \sum_{e\in E_\ell} z_e$; 
      (\proc) \textbf{Negative correlation}: For any $\ell \in I \cup J$, $S \subseteq E_\ell$ and $z \in \{0,1\}$, $\Pr[\wedge_{e \in S} (Z_e=z)] \le \prod_{e \in S} \Pr[Z_e=z] $.

As for the algorithm design in the online settings, we have devised two specific linear programs (LPs) (see \LP~\eqref{LP:1} and \LP~\eqref{LP:2}) to upper bound the performance of the clairvoyant optimal solution (\OPT) for \on and \ons, respectively. Our algorithms follow a straightforward approach: we first solve the LP and extract the marginal distribution among all edges in \OPT, and then utilize it to guide the online actions. The main challenge lies in bounding the performance gap between \OPT and our algorithms. To tackle this challenge, we leverage the technique of swap rounding as introduced by \citet{chekuri2010dependent} and incorporate ideas from contention resolution schemes~\cite{CRS-14} to upper bound the gap between the multilinear relaxation and the concave closure for a monotone submodular function. Additionally, we propose two specific Balls-and-Bins models to facilitate our competitive analysis.

\xhdr{Other Related Work}.
Our offline version can be viewed as a strict special case of maximization of a general monotone submodular function (a sum of weighted coverage functions here) subject to $\ell$-matroids ($\ell=2$ here). For this problem,~\citet{lee10} gave a \emph{local-search} based algorithm that achieves an approx. ratio of $1/(\ell+\epsilon)$ for any $\ell \ge 2, \epsilon>0$.~\citet{Kanthi-17} studied a more general setting which has an intersection of $k$ matroids and a single knapsack constraint.~\citet{karimi2017stochastic} studied a stochastic version of offline  weighted-coverage maximization  but without capacity constraint. 

There is a large body of studies on different variants of submodular maximization problems. Here we list only a few samples that have considered some online arrival \bluee{settings}.~\citet{dickerson2019balancing} studied a variant of online matching model under KIID, whose objective is to maximize a single monotone submodular function over the set of all matched edges. They gave a $0.399$-competitive algorithm.~\citet{nips16Hossein} studied an extension of SWM under adversarial, where each offline agent is associated with two monotone submodular functions.  They proposed a parameterized online algorithm, which is shown to achieve good competitive ratios simultaneously for both objectives. ~\citet{korula2018online} considered SWM under the random arrival order: they showed that Greedy can beat $1/2$, which is the best competitive ratio possibly achieved for SWM under adversarial.~\citet{rawitz2016online} introduced the Online Budgeted Maximum Coverage problem, where we have to select a collection of sets to maximize the coverage, under the constraint that \bluee{the} cost of all sets is at most a given budget. They considered the adversarial arrival setting but under preemption, \ie irrevocable rejections of previously selected sets are \bluee{allowed}.~\citet{anari2017robust} proposed a robust version of online submodular maximization problem, where during each round of \bluee{the} online phase, we need to select a set and then followed by the adversarial arrivals of $k$ monotone submodular functions. The goal is to maximize the sum of min \bluee{values} over the $k$ functions during the whole online phase.  There are several other studies that have considered online maximization of a general submodular function  under the adversarial arrival order; see, \eg \cite{buchbinder-15,chan2017}.

\section{OFFline Capacitated Coverage Maximization (\off)}\label{sec:off}
Throughout this section, we use $\OPT$ to denote \bluee{both} the optimal algorithm and the corresponding performance. For each edge $e \in E$, let $x_e=1$ indicate $e$ is matched in \OPT. Let  $\cF \doteq I \times \cK$ be the collection of all task-feature pairs. Recall that for each node $\ell \in I\cup J$, $E_\ell$ denotes the set of edges incident to $\ell$. For each $f=(i,k) \in \cF$, let $w_f:=w_{i,k}$, and $z_{f}=1$ indicate that feature $k$ is covered for task $i$ in \OPT; let $E_f =\{e=(i,j) \in E_i:  \chi_{jk}=1\}$, the subset of edges incident to $i$ that can help cover feature $k$. Consider the below relaxed linear program (\LP).
\begingroup
\allowdisplaybreaks
\begin{alignat}{2}
 \max &  ~~\sum_{f \in \cF} w_{f} \cdot z_{f}  && \label{obj-1} \\
 & z_{f} \le \min \big(1, \sum_{e \in E_{f}} x_{e}\big)  &&~~  \forall f \in \cF \label{cons:f} \\ 
 &\sum_{e \in E_i} x_{e} \le b_i   &&  \forall i \in I \label{cons:i} \\ 
&\sum_{e \in E_j} x_{e} \le b_j  &&  \forall j \in J \label{cons:j} \\ 
 &  0 \le x_{e},z_{f} \le 1  && ~~ \forall e\in E, f\in \cF. \label{cons:e}
 \end{alignat}
\endgroup
 
\begin{lemma}\label{lem:off-1}
The optimal value of \LP~\eqref{obj-1} is a valid upper bound on the optimal performance for \off.
\end{lemma}
Note that when we require all $x_e, z_{f} \in \{0,1\}$, \LP~\eqref{obj-1} is reduced to an IP program whose optimal value is exactly equal to \OPT. That's why the relaxed  \LP~\eqref{obj-1} offers a valid upper bound for \OPT.  Our DR-based approximation algorithm is formally stated \bluee{below}. 
\vspace{-0.2in}
\begin{algorithm}[ht!] 
\begin{algorithmic}[1]
\DontPrintSemicolon
\STATE Solve \LP~\eqref{obj-1} to get an optimal solution $\x^*=(x^*_{e})$. \;
Apply dependent rounding (DR) in~\cite{gandhi2006dependent} to $\x^*$ and let $\X^*=(X^*_{e})$ be the random binary vector.\;
\STATE Match all edges $e$ with $X^*_{e}=1$. 
\caption{A dependent-rounding (DR) based algorithm for \off (\algoff).}
\label{alg:off}
\end{algorithmic}
\end{algorithm}
\vspace{-0.2in}

Observe that $\X^*=(X^*_{e})$ will be a feasible allocation with probability one. Consider a given $\ell \in I\cup J$, we have  $\sum_{e \in E_\ell} X^*_e \le \lceil \sum_{e\in E_\ell} x^*_e \rceil \le b_\ell$. The first inequality is due to (\prob) from DR, while the second follows from the feasibility of $\x^*$ to \LP~\eqref{obj-1}.

 \begin{proof}[Proof of Theorem~\ref{thm:main-a}.]
We first show that $\E[\algoff] \ge (1-1/\sfe) \LP\eqref{obj-1}$. By Lemma~\ref{lem:off-1}, we have $\E[\algoff] \ge 
 (1-1/\sfe)\OPT$, which establishes the approximation ratio. Consider a given pair $f=(i,k)$, and let $Z_{f}=1$ indicate that feature $k$ is covered for task $i$ in $\algoff$. Observe that $Z_{f}=1$ iff at least one edge $e \in E_{f}$ is matched in \algoff. 
 \begin{align*}
& \E[Z_{f}]=\Pr\big[\vee_{e\in E_{f}} (X^*_{e}=1)\big]=1-\Pr\big[\wedge_{e \in E_{f}} (X^*_{e}=0)\big]\\
 &\ge 1-\prod_{e \in E_{f}} \Pr[X^*_{e}=0]~~\mbox{\big(due to (\proc) of DR\big)} =  1-\prod_{e \in E_{f}} (1-x_{e}^*)~~\mbox{\big(due to (\proa) of DR\big)}.
 \end{align*}
 Let $\alp \doteq \sum_{e \in E_f} x_{e}^*$. Note that $\prod_{e \in E_f} (1-x_{e}^*) \le \sfe^{-\alp}$. Therefore, $\E[Z_{f}] \ge 1-\sfe^{-\alp}$.  Let $z_{f}^*$  be the optimal value on $f$ in \LP~\eqref{obj-1}. The optimality of $z^*_{f}$ suggests that $z^*_{f}=\min(1, \alp)$. Thus,
$\frac{\E[Z_{f}]}{z^*_{f}} \ge \frac{1-\sfe^{-\alp}}{\min(1, \alp)} \ge 1-1/\sfe$,
 where the second inequality is tight when $\alp=1$. Therefore, we claim that $\E[Z_{f}] \ge (1-1/\sfe) \cdot z^*_{f}$ for all $f \in \cF$. By linearity of expectation, we have
\begin{align*}
\E[\algoff] &=\E\big[\sum_{f \in \cF} w_{f}\cdot Z_{f}\big] \ge (1-1/\sfe) \cdot \sum_{f \in \cF}w_{f} \cdot z^*_{f}\ge (1-1/\sfe)\cdot \OPT.
\end{align*} 
Thus, we establish our result. The optimality of the ratio follows from the work~\citet{feige1998threshold}, which showed that unless $P=NP$, the classical Maximum Coverage Problem (MCP) cannot be approximated within a factor better than $1-1/\sfe$, while MCP can be cast as a special case of \off.
  \end{proof}
 
%%%

\section{Online Capacitated Coverage Maximization (\on)}\label{sec:special}
\subsection{A Benchmark Linear Program (\LP) and a DR-Based Sampling Algorithm}\label{sec:lp-1}
We present a specific benchmark LP for \on when every task is associated with a weighted coverage utility function.  For each edge $e=(i,j) \in E$, let $x_{e} \in [0,1]$ the probability that $j$ is assigned to $i$ in a clairvoyant optimal (\OPT). Note that for coverage utility functions, \OPT will have no incentive to assign $j$ to $i$ more than once when worker $j$ has multiple online arrivals. Similar to \off, let $\cF=I \times \cK$ be the collection of all task-feature pairs. For each $f=(i,k) \in \cF$, let $z_{f}$ denote the probability that feature $k$ is covered for task $i$ in \OPT. Recall that for each $f=(i,k)$, $w_f:=w_{i,k}$, and $E_f$ denotes the set of edges incident to $i$ that can help cover feature $k$. Consider the following LP.
\begingroup
\allowdisplaybreaks
\begin{align}
\max  & \sum_{f \in \cF} w_{f} \cdot z_f && \label{LP:1}\\
 & z_{f} \le \min \big(1, \sum_{e \in E_{f}} x_{e}\big),  &&  \forall f \in \cF \label{cons:fa} \\ 
 &\sum_{e \in E_i} x_{e} \le b_i,   &&  \forall i \in I \label{cons:ia} \\ 
&x_e \le r_j, && \forall e\in E_j, \forall j \in J \label{cons:eja}\\
&\sum_{e \in E_j} x_{e} \le b_j \cdot r_j,  &&  \forall j \in J \label{cons:ja} \\ 
 &  0 \le x_{e},z_f \le 1,  &&  \forall e\in E, f\in \cF. \label{cons:ea}
\end{align}
\endgroup
Note that Constraint on each $j$ in \LP~\eqref{LP:1} differs from that in \LP~\eqref{obj-1}. Though the two LPs look similar, they serve essentially different purposes. 
\begin{lemma}\label{lem:benchmark}
The optimal value of \LP~\eqref{LP:1} is a valid upper bound on the expected performance of a clairvoyant optimal (\OPT) for \on. 
\end{lemma}
 \begin{proof}
 By definitions of $\{w_f, z_f\}$, we can verify that Objective~\eqref{LP:1} encodes the exact expected performance of \OPT. It would suffice to justify all constraints in~\LP~\eqref{LP:1} for \OPT. Constraint~\eqref{cons:fa} follows from the definitions of $\{x_{e}\}$ and $\{z_{f}\}$.  As for Constraint~\eqref{cons:ia}: The \bluee{left-hand} side (LHS) is equal to the expected number of workers assigned to task $i$; thus, it should be no larger than $b_i$, which is the matching capacity of task $i$. Constraint~\eqref{cons:eja}: the probability that each given edge $e=(i,j)$ gets assigned ($x_e$) should be no more than that $j$  arrives at least once, which is $1-\sfe^{-r_j}$; thus, $x_e \le 1-\sfe^{-r_j} \le r_j$. For Constraint~\eqref{cons:ja}: Note that the LHS is equal to the expected number of times worker $j$ gets assigned; thus, it should be no more than $r_j \cdot b_j$, which is the expected number of online arrivals of $j$ multiplied by its capacity. The last Constraint~\eqref{cons:ea} is true since $\{x_e,z_f\}$ are all probability values. 
 \end{proof}
 Let $\{x_{e}^*, z_{f}^*\}$ be an optimal solution to the benchmark LP~\eqref{LP:1}. Our DR-based sampling algorithm, denoted by $\algon$, is formally stated in Algorithm~\ref{alg:on}. 
\vspace{-0.2in}
  \begin{algorithm}[ht!]
 \begin{algorithmic}[1]
 \DontPrintSemicolon
\STATE \textbf{Offline Phase}:  \;
\STATE Solve \LP~\eqref{LP:1} to get an optimal solution $\x^*=(x^*_{e})$. \;
\STATE \textbf{Online Phase}: \;
\FOR{$t=1,\ldots,T$}
\STATE Let an online worker (of type) $j$ arrives at time $t$. Apply dependent rounding (DR) in~\cite{gandhi2006dependent} to the vector 
$\{x^*_e/r_j| e\in E_j\}$, and let $\{ X^*_{e}| e \in E_j\}$ be the random binary vector output.\label{alg:step-2a}\;
\STATE Match the edge $e=(i,j) \in E_j$ if $X^*_{e}=1$, $e$ is not matched before, and $i$'s capacity remains.\label{alg:step-2b} \ENDFOR
\caption{A dependent-rounding (DR) based sampling algorithm for \on ($\algon$).}
\label{alg:on}
 \end{algorithmic}
\end{algorithm}
\vspace{-0.2in}

\textbf{Remarks on $\algon$}. (i)  In Step~\eqref{alg:step-2a}, the vector $\{x^*_e/r_j| e\in E_j\}$ satisfies that each entry $x_e^*/r_j \in [0,1]$ due to Constraint~\eqref{cons:eja} and the total sum  of all entries should be no more than $b_j$ from Constraint~\eqref{cons:ja}. (ii) In Step~\eqref{alg:step-2b}, we are guaranteed that the capacity of worker $j$ will never be violated since $\sum_{e\in E_j}X^*_e \le \lceil \sum_{e\in E_j}x^*_e/r_j\rceil \le b_j$ thanks to  (\prob) of DR. (iii) Due to potentially multiple arrivals of worker $j$, edge $e \in E_j$ can be possibly matched before upon the arrival of $j$ at $t$. We can ignore $e$ in this case since it adds nothing for a second match. 
  
 \subsection{Proof of the Unconstrained Task-Capacity Case in Theorem~\ref{thm:main-1}} \label{sec:thm1b}
%We consider a warm-up case of \on without task-capacity constraints. 
 \begin{proof}
In this case, we can simply ignore the concern that if $i$ has reached the capacity in Step~\ref{alg:step-2b} of \algon. Consider a given pair $f=(i,k)\in \cF$, and let $Z_{f}=1$ indicate feature $k$ is covered for task $i$ in \algon. Observe that $Z_f=1$ iff one edge $e=(ij) \in E_f$ arrives (\ie $j$ arrives) and $e$ gets rounded ($X^*_e=1$) at that time. Consider a given time $t$ and assume $Z_f=0$ at (the beginning of) $t$. We see that $Z_f=1$ at the end of $t$ with probability equal to
\begin{align}
\sum_{e=(i,j) \in E_f} \Pr\big[\mbox{$j$ arrives}\big] \cdot\Pr[X_e^*=1] =\sum_{e=(ij) \in E_f} (r_j/T)\cdot (x_e^*/r_j)=\sum_{e \in E_f}x_e^*/T \doteq \alp/T \label{eqn:thm-2},
\end{align}
where $\alp\doteq \sum_{e \in E_f} x_e^*$. Here the first equality on~\eqref{eqn:thm-2} is partially due to (\proa) of DR. Note that $Z_f=0$ iff $Z_f$ never gets updated to $1$ over all the $T$ rounds. Thus, 
\[
\Pr[Z_f=1]=1-\Pr[Z_f=0]=1-(1-\alp/T)^T \ge 1-\sfe^{-\alp}.
\]
Let $z_f^*$ be the optimal value in \LP~\eqref{LP:1} and we see $z_f^*=\min(1,\alp)$. Therefore, 
\[\frac{\E[Z_{f}]}{z^*_{f}} \ge \frac{1-\sfe^{-\alp}}{\min(1, \alp)} \ge 1-1/\sfe,
\]
where the last inequality is tight when $\alp=1$. Thus, we claim that $\E[Z_{f}] \ge (1-1/\sfe) z^*_{f}$ for all $f\in \cF$. By linearity of expectation, we have 
\begin{align*}
\E[\algon] =\E\big[\sum_{f \in \cF} w_{f}Z_{f}\big] \ge (1-1/\sfe) \sum_{f \in \cF}w_{f} z^*_{f}= (1-1/\sfe)\cdot \mbox{\LP-\eqref{LP:1}} \ge (1-1/\sfe) \cdot \OPT,
\end{align*} 
 where \LP-\eqref{LP:1} and \OPT denote the respective optimal values of the linear program and a clairvoyant optimal, and the last inequality follows from Lemma~\ref{lem:benchmark}. Thus, we claim that $\algon$ achieves a CR at least $1-1/\sfe$. The optimality is due to the work~\cite{kapralov2013online}, which showed that no algorithm can beat $1-1/\sfe$ for online SWM under KIID even when all tasks take unweighted coverage functions. 
\end{proof} 
 
 \subsection{Proof of the General Case in Theorem~\ref{thm:main-1}}\label{sec:thm1a}
 We consider \on when each task $i$ is associated with an arbitrary integer capacity $b_i$.
\begin{theorem}\label{thm:main3}
For each $f=(i,k) \in \cF$, let $Z_{f}=1$ indicate that feature $k$ is covered for task $i$ in \algon.
\[
\E[Z_{f}] \ge 0.580 \cdot\min\Big(1, \sum_{e \in E_{f}}x^*_{e} \Big), ~\forall f \in \cF.
\]
\end{theorem}
We defer the proof of Theorem~\ref{thm:main3} to \bluee{the} next section and first give a full proof of Theorem~\ref{thm:main-1} based on Theorem~\ref{thm:main3}.
\begin{proof}
For \bluee{ease} of notation, we use \LP-\eqref{LP:1} to refer to both \bluee{the} LP itself and the corresponding optimal value. By definition of all tasks' utility functions and the linearity of expectation, we have
\begin{align*}
\E[\algon] &=\sum_{f\in \cF}w_{f} \cdot \E[Z_{f}] \ge  0.580  \sum_{f \in \cF}w_{f} \cdot \min\Big(1, \sum_{e \in E_{f}}x^*_{e} \Big) \ge 0.580  \sum_{f \in \cF}w_{f} \cdot z_{f}^*=0.580 \cdot \mbox{\LP-\eqref{LP:1}}.
\end{align*}
 Note that the second inequality above is due to Constraint~\eqref{cons:fa} of~\LP-\eqref{LP:1}. By Lemma~\ref{lem:benchmark},
 we claim that \algon achieves a CR at least $0.580$. Thus, we complete the first part of Theorem~\ref{thm:main-1}. 
\end{proof}
 \subsubsection{Proof of Theorem~\ref{thm:main3}}\label{sec:cov-3}
Focus on a given task-feature pair $f=(i,k) \in \cF$. Observe that feature $k$ will be covered for $i$ iff one edge $e \in E_{f}$ is matched before task $i$ exhausts its capacity $b_i$. Note that during each round $t$, $\algon$ will match an edge $e=(i,j) \in E_{f}$ with probability $\sum_{e=(i,j)\in E_{f}}(r_j/T) \cdot (x_{e}^*/r_j)=\sum_{e \in E_{f}}x_{e}^*/T \doteq p/T$; meanwhile, $\algon$ will match an edge $e \in (E_i-E_{f})$ with probability $\sum_{e \in (E_i-E_{f})} x_{e}^*/T \doteq q/T$. By the nature of \algon, it will keep on sampling edges from $E_{f}$ and $(E_i-E_{f})$ with respective probabilities $p/T$ and $q/T$ during each round until either $b_i$ edges are matched or we reach the last round $T$. Note that from Constraint~\eqref{cons:ia} of \LP-\eqref{LP:1}, we have $p+q \le b_i$. Let $Z=1$ indicate that at least an edge $e\in \cN_{f}$ gets matched by the end. The result in Theorem~\ref{thm:main3} can be equivalently stated as
$\E[Z] \ge 0.580 \cdot \min (1,p)$, where $p:=\sum_{e\in E_f}x_e^*$.

Let us treat the task $i$ as a bin with a capacity $b:=b_i$ and edges from $E_{f}$ and $E_i-E_{f}$ as two types of balls. Then we can restate the question above alternatively as a Balls-and-Bins problem as follows.

\noindent \textbf{A Balls-and-Bins Model (BBM)}.  Suppose we have one single bin and two types of balls, namely type I and type II. We have $T$ rounds and during each round $t \in [T]$, one ball of \bluee{type} I and type II will be sampled with respective probabilities $p/T$ and $q/T$ (with replacement)\footnote{Note that $\algon$ may match multiple different edges in $E_f$, though \bluee{every} single edge will be matched at most once. That's why a ball of type I should be sampled with replacement.}; and with probability $1-(p+q)/T$, no ball will be sampled. Here we assume $T \gg b\ge 1 $ and $0\le p,q \le b$ and $p+q\le b$. The bin has a capacity of $b$ in the way that the sampling process will stop either the bin has $b$ balls (copies will be counted) or we reach the last round $t=T$. Let $Z$ indicate that at least one ball of type I will be sampled by the termination of BBM. We aim to prove that $\E[Z] \ge 0.580 \cdot \min (1,p)$. We split the whole proof into the following two lemmas.

\begin{lemma}\label{lem:cov-1}
If $p \ge 1$, we have $\E[Z] \ge 0.580 $.
\end{lemma}
\begin{lemma}\label{lem:cov-2}
If $p < 1$, we have $\E[Z] \ge 0.580 \cdot p$.
\end{lemma}
%We present a full proof of Lemma~\ref{lem:cov-1} here and defer the proof of Lemma~\ref{lem:cov-2} to Appendix. 
\begin{proof}[Proof of Lemma~\ref{lem:cov-1}]
Suppose $p \ge 1$ and $q \le b-p \le b-1$. To minimize the probability that a ball of type I gets sampled, we can verify that the adversary will arrange $p=1$ and $q=b-1$. Let $A_{I,t}=1$ and $A_{II,t}=1$ indicate a ball of type I and a ball of type II get sampled at $t$, respectively. Let $Z_t=1$ indicate that a ball of type I gets sampled for the first time at time $t$. Suppose $\Bi(\cdot,\cdot)$ and $\Pois(\cdot)$ represent a random variable following a binomial and Poisson distributions, respectively.
\begingroup
\allowdisplaybreaks
\begin{align*}
\E[Z]&=\sum_{t=1}^T \E[Z_t]=\sum_{t=1}^T \E[A_{I,t}] \Pr[A_{I,t'}=0, \forall t'<t] 
 \cdot \Pr\bb{\sum_{t'<t}A_{II,t} \le b-1|  A_{I,t'}=0, \forall t'<t}\\
& =\sum_{t=1}^T \frac{1}{T}\bp{1-\frac{1}{T}}^{t-1}\Pr\bb{\Bi\bp{t-1, \frac{q}{T-1}} \le b-1}.
\end{align*}
\endgroup
The equality \bluee{in} the last line can be justified as follows. Assume $A_{I,t'}=0, \forall t'<t$, which means that a ball of type I never gets sampled before $t$. Conditioning on that, a ball of type II will get sampled during each round $t'<t$ with a probability $(q/T)/(1-p/T)=q/(T-1)$. Therefore, 
\begingroup
\allowdisplaybreaks
\begin{align*}
\E[Z] &\ge \sum_{t=b+1}^T \frac{1}{T}\bp{1-\frac{1}{T}}^{t-1}\Pr\bb{\Bi\bp{t-1, \frac{q}{T-1}} \le b-1}\\ 
&= \sum_{t=b}^{T-1} \frac{1}{T}\bp{1-\frac{1}{T}}^{t} \sum_{\ell=0}^{b-1}{ t \choose \ell } \Big(\frac{q}{T-1} \Big)^\ell \Big(1-\frac{q}{T-1} \Big)^{t-\ell} \\
&=\sum_{\ell=0}^{b-1} \frac{q^\ell}{\ell!}
\sum_{t=b}^{T-1} \frac{1}{T}\bp{1-\frac{1}{T}}^{t} \frac{t\cdot (t-1)\cdots (t-\ell+1)}{(T-1)^\ell}\Big(1-\frac{q}{T-1} \Big)^{t-\ell} \\
&=\sum_{\ell=0}^{b-1} \frac{q^\ell}{\ell!} \int_0^1  \zeta^\ell \sfe^{-(q+1)\zeta} d\zeta -O(1/T) =\int_{0}^1 d\zeta   \sum_{\ell=0}^{b-1}  \frac{(\zeta q)^\ell}{\ell!} \sfe^{-q\zeta} \sfe^{-\zeta} -O(1/T)\\
&=\int_{0}^1 d\zeta  \cdot  \sfe^{-\zeta} \Pr\Big[ \mathrm{Pois}(\zeta q) \le b-1 \Big]-O(1/T)\\
&=\int_{0}^1 d\zeta \cdot  \sfe^{-\zeta} \Pr\Big[ \mathrm{Pois}(\zeta q) \le q \Big]
-O(1/T)
\doteq H(q)-O(1/T).
\end{align*}
\endgroup
Recall that $T \gg 1$ and thus, we can ignore the term of $O(1/T)$. Note that $q=b-1$, which takes integer values only. When $q$ is small, we can use Mathematica to verify that 
\[
\min_{q \in \{0,1,\ldots,100 \}} H(q)=H(2)=\frac{1}{27}\Big(19-\frac{67}{\sfe^3}\Big)\sim 0.580.
\]
Now we try to lower bound $H(q)$ for large $q$ values. Applying the upper tail bound of a Poisson random variable due to the work of~\cite{pois-tail}, we have $\Pr[\Pois(\zeta q) >q] \le \exp \bp{ \frac{-q(1-\zeta)^2}{2}}$. 
Therefore,
\begin{align*}
H(q) & =\int_{0}^1 d\zeta \cdot  \sfe^{-\zeta} \Pr\Big[ \mathrm{Pois}(\zeta q) \le q \Big] \ge\int_{0}^1 d\zeta \cdot  \sfe^{-\zeta} \bP{1-\exp \bp{ \frac{-q(1-\zeta)^2}{2}}} \doteq H_L(q).
\end{align*}
We can verify that (i) $H_L(q)$ is an increasing function of $q$ when $q>0$; and (ii) $H_L(100) \ge 0.582$. Thus, we have $H(q) \ge H_L(q) \ge 0.582$ for all $q >100$. Finally, we conclude that $\E[Z] \ge H(q) \ge 0.580$ for all non-negative integer $q$. This establishes the final result.
\end{proof}

\begin{comment}
Recall that we have defined a Balls-and-Bins Model in Section~\ref{sec:cov-3}. Let $Z=1$ indicate that at least one ball of type I gets sampled by the termination of BBM. We aim to prove that $\E[Z] \ge 0.580 \cdot \min (1,p)$. Lemma~\ref{lem:cov-2} states that
\[
\E[Z] \ge 0.580 \cdot p, ~\forall~0  \le p<1.
\]
The proof of Lemma~\ref{lem:cov-2} is similar to that of Lemma~\ref{lem:cov-1}.
\end{comment}
%\xhdr{Proof of Lemma~\ref{lem:cov-2}}.
\begin{proof}[Proof of Lemma~\ref{lem:cov-2}]
Recall that $A_{I,t}$ and $A_{II,t}$ indicate a ball of type I and II gets sampled at $t$, respectively, and $Z_t=1$ \bluee{indicates} that a ball of type I gets sampled for the first time at time $t$. WLOG assume $q=b-p$. Thus, we have
\begingroup
\allowdisplaybreaks
\begin{align*}
\E[Z]&=\sum_{t=1}^T \E[Z_t]\\
&=\sum_{t=1}^T \Pr[A_{I,t}] \Pr[A_{I,t'}=0, \forall t'<t] \Pr\bb{\sum_{t'<t}A_{II,t} \le b-1|  A_{I,t'}=0, \forall t'<t}\\
& =\sum_{t=1}^T \frac{p}{T}\bp{1-\frac{p}{T}}^{t-1}\Pr\bb{\Bi\bp{t-1, \frac{q}{T-p}} \le b-1}\\
&\ge \sum_{t=b+1}^T \frac{p}{T}\bp{1-\frac{p}{T}}^{t-1} \sum_{\ell=0}^{b-1}{ t-1 \choose \ell } \Big(\frac{q}{T-p} \Big)^\ell \Big(1-\frac{q}{T-p} \Big)^{t-1-\ell} \\
&=p \cdot \sum_{\ell=0}^{b-1} \frac{q^\ell}{\ell!}
\sum_{t=b+1}^T \frac{1}{T}\bp{1-\frac{p}{T}}^{t-1} \frac{(t-1)\cdots (t-\ell)}{(T-p)^\ell}\Big(1-\frac{q}{T-p} \Big)^{t-1-\ell} \\
&=p \cdot \sum_{\ell=0}^{b-1} \frac{q^\ell}{\ell!} \int_0^1  \zeta^\ell \cdot \sfe^{-p \zeta} \cdot \sfe^{-q \zeta} d\zeta  ~~~\mbox{\bp{by taking $T \rightarrow \infty$ since $p,q,b \ll T$}}
\\
&\ge p \cdot \sum_{\ell=0}^{b-1} \frac{(b-1)^\ell}{\ell!} \int_0^1  \zeta^\ell \cdot \sfe^{-(b-1) \zeta} \cdot \sfe^{- \zeta} d\zeta  ~~~\bp{q=b-p>b-1, p+q=b}\\
&=p \cdot \int_{0}^1 d\zeta  \cdot \sum_{\ell=0}^{b-1}  \frac{( (b-1)\zeta)^\ell}{\ell!}\cdot \sfe^{-(b-1)\zeta} \cdot \sfe^{-\zeta}\\
&=p \cdot \int_{0}^1 d\zeta \cdot  \sfe^{-\zeta} \Pr\Big[ \mathrm{Pois}\bp{(b-1)\zeta} \le b-1 \Big] \doteq p \cdot H(b-1).
\end{align*}
\endgroup
Note that the function $H(\cdot)$ shares the same form as that in the proof of Lemma~\ref{lem:cov-1}. Thus, applying the same analysis as before, we claim that 
$\E[Z] \ge 0.580 \cdot p$ for any $0 \le p<1$.
\end{proof}

\section{Online Capacitated Submodular Maximization (\ons)} \label{sec:gen}
For the general case, let's recall that each task $i$ is associated with a general monotone submodular utility function $g_i$ over the \bluee{ground set} of ${\cN_i}$ (the set of all neighbors of $i$). Here are two remarks. First, under KIID, it is possible to have multiple online arrivals for some worker types. Therefore, in the general case, we need to define $g_i$ over all possible \emph{multisets} of $\cN_i$. However, we can skip this by introducing multiple copies of worker types with high arrival rates. This ensures that, with high probability, each worker type has at most one online arrival. As a result, we effectively treat multiple arrivals of a single worker type as the arrivals of multiple distinct worker types. Second, inspired by studies~\cite{DBLP:journals/corr/abs-2002-10697,Faez-17}, we assume that each task $i$ has a \emph{constant} matching capacity $b_i$. This assumption is motivated by practical gig economy platforms, where the capacity of each task $i$ (\ie the maximum number of workers that can be allocated to $i$) is typically small due to its finite number of copies and/or limited budget.
%Please note that in most practical cases, the capacity of each task is relatively small due to its finite number of copies and/or limited budget.

In the following, we present a configuration-LP based algorithm. Recall that for each task $i$ (worker $j$),  $\cN_i$ ($\cN_j$) is the set of neighbors  incident to $i$ ($j$) in the input graph $G=(I,J,E)$. Let $\Lam_i$ be the collection of all subsets of neighbors of $i$ with cardinality no larger than $b_i$, \ie $\Lam_i=\{S \subseteq \cN_i, |S| \le b_i\}$.  Since each task $i$ has a capacity $b_i$, we claim that any clairvoyant optimal (\OPT) will select at most one subset $S \in \Lam_i$ and assign it to task $i$ for every $i \in I$. Let $x_{i,S}$ be the probability that $S \in \Lam_i$ is assigned to $i$ in \OPT. For each $ e=(i,j) \in E$, let $\Lam_{ij}  =\{S \in \Lam_i: S \ni j \}$, which denotes the collection of subsets of neighbors of $i$ that include worker $j$ but have no more than $b_i$ workers in total. Consider the LP below.
\begingroup
\allowdisplaybreaks
\begin{align}
\max  & \sum_{i \in I} \sum_{S \in \Lam_i} g_i(S) \cdot x_{i,S}  &&   \label{LP:2}\\
& \sum_{S \in \Lam_i} x_{i,S} \le 1, &&\forall i \in I \label{cons:i2}\\
&\sum_{S \in \Lam_{ij}}x_{i, S} \le r_j, && \forall j\in J, (ij)\in E \label{cons:je2}\\
& \sum_{i\in \cN_j}\sum_{S \in \Lam_{ij}}x_{i, S} \le r_j \cdot b_j, &&\forall j \in J \label{cons:j2} \\
& 0 \le x_{i,S}, && \forall i \in I, S \in \Lam_i. \label{cons:S2}
\end{align}
\endgroup
Observe that \LP~\eqref{LP:2} can be constructed and solved within polynomial time since $\{b_i\}$ are assumed constants. This is the only place we need the constant-task-capacity assumption.

\begin{lemma}\label{lem:benchmark-2}
The optimal value of \LP~\eqref{LP:2} is a valid upper bound on the expected performance of a clairvoyant optimal (\OPT) for \ons.
\end{lemma}
\begin{proof}
 Note that the objective function captures the expected utilities of a clairvoyant optimal (\OPT). It would suffice to show the feasibility of all constraints in \LP-\eqref{LP:2}. Constraint~\eqref{cons:i2} is valid since \OPT will select at most one subset  $S \in \Lam_i$ and assign it to task $i$ for each $i$. Constraint~\eqref{cons:je2}: The LHS represents the marginal probability that $j$ is assigned to $i$ in $\OPT$, which should be no more than the probability that $j$ arrives at least once that is equal to $1-\sfe^{-r_j} \le r_j$.  
 For Constraint~\eqref{cons:j2}: The LHS is equal to the expected number of times that $j$ is assigned; thus, it should be no larger than its online arrivals multiplied by its capacity, which is $r_j \cdot b_j$. 
\end{proof}
 
Let $\{x_{i,S}^*\}$ be an optimal solution to \LP-\eqref{LP:2}. For each edge $e=(i,j) \in E$, let $y^*_{e} \doteq \sum_{S \in \Lam_{ij}} x^*_{i,S}$, which can be interpreted as the marginal probability that worker $j$ is assigned to task $i$ in a clairvoyant optimal (\OPT). From the feasibility of $\{x_{i,S}^*\}$ to \LP-\eqref{LP:2}, we claim the following properties of $\{y_e^*\}$. Recall that $E_i$ ($E_j$) is the set of edges incident to $i$ ($j$). 
\begin{lemma}\label{eqn:y}
(1) $y^*_{e} \le r_j, \forall e=(ij) \in E$; (2) $\sum_{e \in E_j} y^*_{e} \le r_j \cdot b_j, ~\forall j \in J$; (3)$\sum_{e \in E_i} y_{e}^* \le b_i, \forall i \in I$.
\end{lemma}
 \begin{proof}
Consider a given edge $e=(i,j)$. The first inequality is due to Constraint~\eqref{cons:je2}; the second is due to $\sum_{e \in E_j} y^*_{e}=\sum_{i \in \cN_j} \sum_{S \in \Lam_{ij}} x_{i,S}^* \le r_j \cdot b_j$; the last is valid since $\sum_{e \in E_i}y_e^*=\sum_{j \in \cN_i} y_{i,j}^* =\sum_{j \in \cN_i} \sum_{S \in \Lam_{ij}}x_{i,S}^*=\sum_{S \in \Lam_i} |S| \cdot x^*_{i,S} \le \sum_{S \in \Lam_i} b_i \cdot x^*_{i,S} \le b_i$.
\end{proof}
Based on the solution $\{y_{e}^*\}$, our sampling algorithm is formally stated as below. 
\vspace{-0.2in}
\begin{algorithm}[ht!] 
\begin{algorithmic}[1]
\DontPrintSemicolon
\STATE \textbf{Offline Phase}: \;
\STATE Solve LP~\eqref{LP:2} for an optimal solution $\{x_{i,S}^*\}$. For each $e=(i,j)$, let $y_e^* \doteq \sum_{S \in \Lam_{ij}}x_{i,S}^*$. \;
\STATE \textbf{Online Phase}: \;
\FOR{$t=1,2,\ldots,T$}
	\STATE Let an online worker (of type) $j$ arrive at time $t$. \;
		\STATE Apply dependent rounding (DR) in~\cite{gandhi2006dependent} to the vector 
$\{y^*_{e}/r_j| e\in E_j\}$, and let $\{ Y^*_{e}| e \in E_j\}$ be the random binary vector output. \label{alg:step-3a}\;
	\STATE Match the edge $e=(i,j) \in E_j$ if $Y^*_{e}=1$ and $i$'s capacity remains. \label{alg:step-3b}
\ENDFOR
\caption{Dependent-Rounding-Based Algorithm for \ons ($\nna$).} 
\label{alg3}
\end{algorithmic}
\end{algorithm}
\vspace{-0.2in}

\textbf{Remarks on $\nna$}. (1) in Step~\eqref{alg:step-3a} of $\nna$, the vector $\{y^*_{e}/r_j| e\in E_j\}$ has each entry $y_e^*/r_j \in [0,1]$ and the total sum $\sum_{e\in E_j} y_e^*/r_j \le b_j$, due to the first and second inequalities in Lemma~\ref{eqn:y}; and (2) in Step~\eqref{alg:step-3b} of $\nna$, we are guaranteed that the capacity of worker $j$ will never be violated since $\sum_{e\in E_j}Y^*_e \le \lceil \sum_{e\in E_j}y^*_e/r_j\rceil \le b_j$ thanks to (\prob) of DR.

 %\subsection{A proof sketch of the main theorem~\ref{thm:main-2}}
  \subsection{Proof  of Theorem~\ref{thm:main-2}}
 \label{sec:proof_2}
\textbf{A Second Balls-and-Bins Model} (\bbsp): Consider a given task $i \in I$. Similar to the previous analysis, we treat task $i$ as a bin with capacity $b_i$, and each worker $j \in \cN_i$ as a ball (of type) $j$. \nna implies that during each round $t \in [T]$, a ball $j \in \cN_i$ will be sampled (or arrive) with probability $ (r_j/T) \cdot (y_{ij}^*/r_j)=y_{ij}^*/T$ (with replacement), and no ball will arrive with probability $1-\sum_{e \in E_i} y_{e}^*/T$. Also, the sampling process will stop when either a total of $b_i$ balls arrive (copies will be counted) or we reach the last round $t=T$.

Let $b_i=b$. By Lemma~\eqref{eqn:y}, we have  $\sum_{e \in E_i}  y_{e}^*\le b$. WLOG assume that $\sum_{e \in E_i}  y_{e}^*\doteq y_i^*=b$. We can verify that the Worst Scenario (WS) will arrive when $y_i^*=b$ since \bbsp will terminate faster if not the same, compared to the case $y_i^*<b$. Let $\cS$ denote the (random) set of balls that arrive  before termination. Thus, the expected utility obtained on task $i$ in \nna should be $\E[g_i(\cS)]$. Let $A=|\cS|$, which denotes the (random) number of arrivals of balls before  termination. Observe that $A=\min(b, \Pois(b))$, a truncated Poisson random variable with $\Pois(b)$ denoting a Poisson random variable of mean $b$. Let $\OPT_i \doteq \sum_{S \in \Lam_i} g_i(S) \cdot x_{i,S}^*$, which denotes the expected utility  on task $i$ in a clairvoyant optimal (\OPT). We aim to show that $\E[g_i(\cS)] \ge 0.436 \cdot \OPT_i$ for each $i \in I$, and thus, by linearity of expectation we establish the final competitive ratio (CR) of \nna.  Here are two key ingredients to prove Theorem~\ref{thm:main-2}. 
   \begin{theorem}\label{thm:alg2}
If $A=1$, we have  $\E[g_i(\cS) |A=1] \ge \frac{\OPT_i}{b}$. If $2 \le A=\ell \le b$, we have 
 \[
  \E[g_i(\cS) |A=\ell] \ge \OPT_i \cdot \big(1-\sfe^{-1+(1-\frac{1}{b})^\ell}\big). 
    \]
   \end{theorem}
   
      \begin{lemma}\label{lem:alg2-d}
  Let $\phi (b,\ell) \doteq 1-\sfe^{-1+(1-\frac{1}{b})^\ell}$ and $\Phi(b)\doteq \frac{\Pr[\Pois(b)=1]}{b} + \sum_{2 \le \ell<b}\Pr[\Pois(b)=\ell] \cdot \phi(b,\ell)+  \Pr[\Pois(b) \ge b] \cdot \phi(b,b)$, where $\Pois(b)$ denotes a Poisson random variable with mean $b$. We have that $\Phi(b) \ge 0.436$ when $b \ge 2$.
   \end{lemma}
 %  Proofs of Theorem~\ref{thm:alg2} and Lemma~\ref{lem:alg2-d} are referred to Appendix; see Sections~\ref{sec:alg2} and~\ref{sec:lem-alg2}. Now we show how it implies Theorem~\ref{thm:main-2}.
 In the following, we show how the two results above lead to Theorem~\ref{thm:main-2}. We defer the proof of Theorem~\ref{thm:alg2} to Appendix~\ref{sec:alg2} and the proof of Lemma~\ref{lem:alg2-d} afterwards.
   \begin{proof}[Proof of Theorem~\ref{thm:main-2}]
   For the case $b=1$, we see that $A \le 1$, and $\Pr[A=1]=\Pr[\Pois(1) \ge 1]=1-1/\sfe$. Therefore, 
   \[
   \E[g_i(S) ]=\E[g_i(S) |A=1] \Pr[A=1] \ge \OPT_i \cdot (1-1/\sfe).
   \]
  The inequality above is partially due to Theorem~\ref{thm:alg2} when $A=1$. Now consider a general case $b \ge 2$. Recall that $\phi(b,\ell) \doteq 1-\sfe^{-1+(1-\frac{1}{b})^\ell}$. From Theorem~\ref{thm:alg2}, we have that 
   \begin{align*}
   \E[g_i(\cS)] &=\E[g_i(S) |A=1] \Pr[A=1]+\sum_{2\le \ell \le b}\E[g_i(S) |A=\ell] \Pr[A=\ell] \\
   & \ge  \frac{\OPT_i}{b} \cdot \Pr[\Pois(b)=1]
   + \OPT_i \cdot \bp{\sum_{ \ell=2}^{b-1}\phi(b,\ell)\cdot \Pr[\Pois(b)=\ell]+ \phi(b,b)\cdot \Pr[\Pois(b) \ge b]} \\
   &=\OPT_i \Big(\frac{\Pr[\Pois(b)=1]}{b}    + \sum_{\ell=2}^{b-1}\Pr[\Pois(b)=\ell] \cdot \phi(b,\ell)+  \Pr[\Pois(b) \ge b] \cdot \phi(b,b)\Big)\\
   & \doteq \OPT_i \cdot \Phi(b) \ge  \OPT_i \cdot 0.436.  ~~\mbox{\sbp{due to Lemma~\ref{lem:alg2-d}}}
   \end{align*}
     \end{proof}

%\subsection{Proof of Lemma~\ref{lem:alg2-d}}\label{sec:lem-alg2}
   \begin{proof}[Proof of Lemma~\ref{lem:alg2-d}]
Observe that for a Poisson random variable $\Pois(b)$, its median number is $b$ for any integral value of $b$.\footnote{https://math.stackexchange.com/questions/455054/poisson-distribution-and-median.} Thus, we have that for the third part of $\Phi(b)$,
\[ \Pr[\Pois(b) \ge b] \cdot \phi(b,b) \ge \frac{1}{2} \bp{1-\sfe^{-1+1/\sfe}}\ge 0.234.
\]

The second part of $\Phi(b)$ can be rewritten as $H_1-H_2$, where 
\[H_1= \sum_{\ell=2}^{b-1} \Pr[\Pois(b)=\ell] \ge \frac{1}{2}-\sfe^{-b}b-\sfe^{-b}b^b/b!\ge  \frac{1}{2}-\frac{b}{\sfe^b}-\frac{1}{\sqrt{2\pi b}},
 \] 
 and
 \begingroup
\allowdisplaybreaks
 \begin{align}
 H_2& =\sum_{\ell=2}^{b-1} \Pr[\Pois(b)=\ell] \exp\big(-1+(1-1/b)^\ell\big) \nonumber\\
 & \le  \sum_{\ell=2}^{b-1} \Pr[\Pois(b)=\ell] \bp{(1-1/b)^\ell+\frac{1}{2}\big(-1+(1-1/b)^\ell\big)^2} \label{ineq:G-1}\\
 & = \frac{1}{2}\sum_{\ell=2}^{b-1} \Pr[\Pois(b)=\ell]+\frac{1}{2}\sum_{\ell=2}^{b-1} \Pr[\Pois(b)=\ell] (1-1/b)^{2\ell} \nonumber \\
 & \le  \frac{1}{4}+\frac{1}{2}\sum_{\ell=2}^{b-1} \frac{\sfe^{-b}\cdot b^\ell}{\ell!}(1-1/b)^{2\ell} \label{ineq:G-2}
 \\
 & = \frac{1}{4}+\frac{1}{2}\sum_{\ell=2}^{b-1}  \big(b \cdot(1-1/b)^2\big)^\ell \sfe^{-b \cdot(1-1/b)^2}\cdot \frac{1}{\ell!} \cdot \sfe^{b \cdot(1-1/b)^2-b} \nonumber \\
 &\le  \frac{1}{4}+\frac{1}{2} \bp{\frac{1}{2}+\Pr[\Pois(b \cdot(1-1/b)^2)=b-2]+\Pr[\Pois(b \cdot(1-1/b)^2)=b-1]} \cdot\sfe^{-2+1/b} \nonumber\\
 & \le \frac{1}{4}+\frac{\sfe^{-2+1/b}}{4}\bp{1+4\Pr[\Pois(b \cdot(1-1/b)^2)=b-2]}\nonumber\\
 & \le \frac{1}{4}+\frac{\sfe^{-2+1/b}}{4}+ \frac{\sfe^{-2+1/b}}{\sqrt{2\pi (b-2)}}.\nonumber
  \end{align}
  \endgroup
Inequality~\eqref{ineq:G-1} is due to the fact that $\sfe^x \le 1+x+x^2/2$ for all $x \in [-1,0]$. Inequality~\eqref{ineq:G-2} is due to the fact that $b$ is the median number of $\Pois(b)$ for integral $b$. Note that for $\Pois(b \cdot(1-1/b)^2)$, its median number and its mode are both  $b-2$, where the mode number is the value such that its PDF gets the maximum value. Thus, we claim that
\[
\Phi(b) \ge \frac{1}{2} \bp{1-\sfe^{-1+1/\sfe}}+H_1-H_2 \ge \frac{1}{2} \bp{1-\sfe^{-1+1/\sfe}}+ \frac{1}{4}\bp{1-\sfe^{-2+1/b}}-\frac{1+\sfe^{-2+1/b}}{\sqrt{2\pi (b-2)}}-\frac{b}{\sfe^b} \doteq \tau(b).
\]
Observe that the rightmost expression $\tau(b)$ is an increasing function of $b$ when $b \ge 3$. We can verify that (1) when $2 \le b \le 1000$, $\Phi(b)$ takes its minimum value of $0.436$ at $b=4$; (2) When $b \ge 1000$, $\Phi(b) \ge \tau(b)\ge \tau(1000)=0.436$.  Thus, we establish the claim that $\Phi(b) \ge 0.436$ when $b \ge 2$. 
   \end{proof}

\section{Conclusions and Future Work}\label{sec:con}	
In this paper, we have proposed three generic models of capacitated submodular maximization inspired by practical gig platforms. Our models feature the association of each task with either a coverage or a general monotone submodular utility function. We have  presented specific LP-based rounding algorithms for each of the three models and conducted  related approximation-ratio or competitive-ratio analysis. In the following, we discuss a few potential future directions. First, we can explore generalizing the capacity constraint to something more general, such as a matroid. Second, it would be interesting to refine the upper bounds (or establish hardness results) for the benchmark LPs in \bluee{online} settings. Can we narrow or close the gap between the upper and lower bounds in terms of the competitive ratio, similar to what has been achieved in the offline setting?

 \clearpage	
\begin{small}
\bibliographystyle{unsrtnat}
\bibliography{stable_ref}

\begin{thebibliography}{44}
\providecommand{\natexlab}[1]{#1}
\providecommand{\url}[1]{\texttt{#1}}
\expandafter\ifx\csname urlstyle\endcsname\relax
  \providecommand{\doi}[1]{doi: #1}\else
  \providecommand{\doi}{doi: \begingroup \urlstyle{rm}\Url}\fi

\bibitem[Ho and Vaughan(2012)]{ho2012online}
Chien-Ju Ho and Jennifer~Wortman Vaughan.
\newblock Online task assignment in crowdsourcing markets.
\newblock In \emph{Twenty-sixth AAAI conference on artificial intelligence},
  2012.

\bibitem[Assadi et~al.(2015)Assadi, Hsu, and Jabbari]{assadi2015online}
Sepehr Assadi, Justin Hsu, and Shahin Jabbari.
\newblock Online assignment of heterogeneous tasks in crowdsourcing markets.
\newblock In \emph{Third AAAI Conference on Human Computation and
  Crowdsourcing}, 2015.

\bibitem[Xu et~al.(2017)Xu, Srinivasan, Sarpatwar, and Wu]{xu2017budgeted}
Pan Xu, Aravind Srinivasan, Kanthi~K Sarpatwar, and Kun-Lung Wu.
\newblock Budgeted online assignment in crowdsourcing markets: theory and
  practice.
\newblock In \emph{Proceedings of the 16th Conference on Autonomous Agents and
  MultiAgent Systems}, pages 1763--1765. International Foundation for
  Autonomous Agents and Multiagent Systems, 2017.

\bibitem[Ahmed et~al.(2017{\natexlab{a}})Ahmed, Dickerson, and
  Fuge]{ahmed2017diverse}
Faez Ahmed, John~P Dickerson, and Mark Fuge.
\newblock Diverse weighted bipartite b-matching.
\newblock \emph{arXiv preprint arXiv:1702.07134}, 2017{\natexlab{a}}.

\bibitem[Barnab{\`o} et~al.(2019)Barnab{\`o}, Fazzone, Leonardi, and
  Schwiegelshohn]{barnabo2019algorithms}
Giorgio Barnab{\`o}, Adriano Fazzone, Stefano Leonardi, and Chris
  Schwiegelshohn.
\newblock Algorithms for fair team formation in online labour marketplaces.
\newblock In \emph{Companion Proceedings of The 2019 World Wide Web
  Conference}, pages 484--490, 2019.

\bibitem[Cheng et~al.(2016)Cheng, Lian, Chen, Han, and Zhao]{cheng2016task}
Peng Cheng, Xiang Lian, Lei Chen, Jinsong Han, and Jizhong Zhao.
\newblock Task assignment on multi-skill oriented spatial crowdsourcing.
\newblock \emph{IEEE Transactions on Knowledge and Data Engineering},
  28\penalty0 (8):\penalty0 2201--2215, 2016.

\bibitem[Anagnostopoulos et~al.(2012)Anagnostopoulos, Becchetti, Castillo,
  Gionis, and Leonardi]{anagnostopoulos2012online}
Aris Anagnostopoulos, Luca Becchetti, Carlos Castillo, Aristides Gionis, and
  Stefano Leonardi.
\newblock Online team formation in social networks.
\newblock In \emph{Proceedings of the 21st international conference on World
  Wide Web}, pages 839--848, 2012.

\bibitem[Novorol(2018)]{fair-scm}
Claire Novorol.
\newblock \url{https://ai-med.io/ai-biases-ada-health-diversity-women/}, 2018.
\newblock Accessed: 2019-09-20.

\bibitem[Ahmed et~al.(2020)Ahmed, Dickerson, and
  Fuge]{DBLP:journals/corr/abs-2002-10697}
Faez Ahmed, John Dickerson, and Mark Fuge.
\newblock Forming diverse teams from sequentially arriving people.
\newblock \emph{CoRR}, abs/2002.10697, 2020.

\bibitem[Yao et~al.(2018)Yao, Wu, Ke, Tang, Jia, Lu, Gong, Ye, and
  Li]{Yao2018deep}
Huaxiu Yao, Fei Wu, Jintao Ke, Xianfeng Tang, Yitian Jia, Siyu Lu, Pinghua
  Gong, Jieping Ye, and Zhenhui Li.
\newblock Deep multi-view spatial-temporal network for taxi demand prediction.
\newblock AAAI '18, pages 2588--2595, 2018.

\bibitem[Li et~al.(2018)Li, Fu, Wang, Shahabi, Ye, and
  Liu]{DBLP:conf/kdd/LiFWSYL18}
Yaguang Li, Kun Fu, Zheng Wang, Cyrus Shahabi, Jieping Ye, and Yan Liu.
\newblock Multi-task representation learning for travel time estimation.
\newblock KDD '18, pages 1695--1704, 2018.

\bibitem[Xie et~al.(2017)Xie, Wang, Cui, Yang, and Li]{xie2017cqm}
Miao Xie, Qing Wang, Qiang Cui, Guowei Yang, and Mingshu Li.
\newblock Cqm: coverage-constrained quality maximization in crowdsourcing test.
\newblock In \emph{2017 IEEE/ACM 39th International Conference on Software
  Engineering Companion (ICSE-C)}, pages 192--194. IEEE, 2017.

\bibitem[Puthiya~Parambath et~al.(2016)Puthiya~Parambath, Usunier, and
  Grandvalet]{puthiya2016coverage}
Shameem~A Puthiya~Parambath, Nicolas Usunier, and Yves Grandvalet.
\newblock A coverage-based approach to recommendation diversity on similarity
  graph.
\newblock In \emph{Proceedings of the 10th ACM Conference on Recommender
  Systems}, pages 15--22, 2016.

\bibitem[Ge et~al.(2010)Ge, Delgado-Battenfeld, and Jannach]{ge2010beyond}
Mouzhi Ge, Carla Delgado-Battenfeld, and Dietmar Jannach.
\newblock Beyond accuracy: evaluating recommender systems by coverage and
  serendipity.
\newblock In \emph{Proceedings of the fourth ACM conference on Recommender
  systems}, pages 257--260, 2010.

\bibitem[Abbassi et~al.(2013)Abbassi, Mirrokni, and
  Thakur]{abbassi2013diversity}
Zeinab Abbassi, Vahab~S Mirrokni, and Mayur Thakur.
\newblock Diversity maximization under matroid constraints.
\newblock In \emph{Proceedings of the 19th ACM SIGKDD international conference
  on Knowledge discovery and data mining}, pages 32--40, 2013.

\bibitem[Brubach et~al.(2020)Brubach, Sankararaman, Srinivasan, and
  Xu]{brubach2020online}
Brian Brubach, Karthik~Abinav Sankararaman, Aravind Srinivasan, and Pan Xu.
\newblock Online stochastic matching: New algorithms and bounds.
\newblock \emph{Algorithmica}, pages 1--47, 2020.

\bibitem[Jaillet and Lu(2013)]{bib:Jaillet}
Patrick Jaillet and Xin Lu.
\newblock Online stochastic matching: New algorithms with better bounds.
\newblock \emph{Mathematics of Operations Research}, 39\penalty0 (3):\penalty0
  624--646, 2013.

\bibitem[Manshadi et~al.(2012)Manshadi, Gharan, and Saberi]{bib:Manshadi}
Vahideh~H Manshadi, Shayan~Oveis Gharan, and Amin Saberi.
\newblock Online stochastic matching: Online actions based on offline
  statistics.
\newblock \emph{Mathematics of Operations Research}, 37\penalty0 (4):\penalty0
  559--573, 2012.

\bibitem[Haeupler et~al.(2011)Haeupler, Mirrokni, and
  Zadimoghaddam]{bib:Haeupler}
Bernhard Haeupler, Vahab~S. Mirrokni, and Morteza Zadimoghaddam.
\newblock Online stochastic weighted matching: Improved approximation
  algorithms.
\newblock In \emph{Internet and Network Economics}, volume 7090 of
  \emph{Lecture Notes in Computer Science}, pages 170--181. Springer Berlin
  Heidelberg, 2011.
\newblock ISBN 978-3-642-25509-0.

\bibitem[Feldman et~al.(2009)Feldman, Mehta, Mirrokni, and
  Muthukrishnan]{bib:Feldman}
Jon Feldman, Aranyak Mehta, Vahab Mirrokni, and S~Muthukrishnan.
\newblock Online stochastic matching: Beating 1-1/e.
\newblock In \emph{Foundations of Computer Science, 2009. FOCS'09. 50th Annual
  IEEE Symposium on}, pages 117--126. IEEE, 2009.

\bibitem[Vondr{\'a}k(2008)]{vondrak2008optimal}
Jan Vondr{\'a}k.
\newblock Optimal approximation for the submodular welfare problem in the value
  oracle model.
\newblock In \emph{Proceedings of the fortieth annual ACM symposium on Theory
  of computing}, pages 67--74, 2008.

\bibitem[Kapralov et~al.(2013)Kapralov, Post, and
  Vondr{\'a}k]{kapralov2013online}
Michael Kapralov, Ian Post, and Jan Vondr{\'a}k.
\newblock Online submodular welfare maximization: Greedy is optimal.
\newblock In \emph{SODA}, 2013.

\bibitem[Feige(1998)]{feige1998threshold}
U.~Feige.
\newblock A threshold of ln n for approximating set cover.
\newblock \emph{J. ACM}, 45:\penalty0 634--652, 1998.

\bibitem[Lee et~al.(2010)Lee, Sviridenko, and Vondr{\'a}k]{lee10}
Jon Lee, Maxim Sviridenko, and Jan Vondr{\'a}k.
\newblock Submodular maximization over multiple matroids via generalized
  exchange properties.
\newblock \emph{Mathematics of Operations Research (MoR)}, 2010.

\bibitem[Ahmed et~al.(2017{\natexlab{b}})Ahmed, Dickerson, and Fuge]{Faez-17}
Faez Ahmed, John~P Dickerson, and Mark Fuge.
\newblock Diverse weighted bipartite b-matching.
\newblock In \emph{Proceedings of the 26th International Joint Conference on
  Artificial Intelligence}, pages 35--41, 2017{\natexlab{b}}.

\bibitem[Gandhi et~al.(2006)Gandhi, Khuller, Parthasarathy, and
  Srinivasan]{gandhi2006dependent}
Rajiv Gandhi, Samir Khuller, Srinivasan Parthasarathy, and Aravind Srinivasan.
\newblock Dependent rounding and its applications to approximation algorithms.
\newblock \emph{Journal of the ACM (JACM)}, 53\penalty0 (3):\penalty0 324--360,
  2006.

\bibitem[Chekuri et~al.(2010)Chekuri, Vondrak, and
  Zenklusen]{chekuri2010dependent}
Chandra Chekuri, Jan Vondrak, and Rico Zenklusen.
\newblock Dependent randomized rounding via exchange properties of
  combinatorial structures.
\newblock In \emph{2010 IEEE 51st Annual Symposium on Foundations of Computer
  Science}, pages 575--584. IEEE, 2010.

\bibitem[Chekuri et~al.(2011)Chekuri, Vondr{\'{a}}k, and Zenklusen]{CRS-14}
Chandra Chekuri, Jan Vondr{\'{a}}k, and Rico Zenklusen.
\newblock Submodular function maximization via the multilinear relaxation and
  contention resolution schemes.
\newblock \emph{CoRR}, abs/1105.4593, 2011.
\newblock URL \url{http://arxiv.org/abs/1105.4593}.

\bibitem[Sarpatwar et~al.(2017)Sarpatwar, Schieber, and Shachnai]{Kanthi-17}
Kanthi~K Sarpatwar, Baruch Schieber, and Hadas Shachnai.
\newblock Interleaved algorithms for constrained submodular function
  maximization.
\newblock \emph{arXiv preprint arXiv:1705.06319}, 2017.

\bibitem[Karimi et~al.(2017)Karimi, Lucic, Hassani, and
  Krause]{karimi2017stochastic}
Mohammad Karimi, Mario Lucic, Hamed Hassani, and Andreas Krause.
\newblock Stochastic submodular maximization: The case of coverage functions.
\newblock In \emph{Advances in Neural Information Processing Systems}, pages
  6853--6863, 2017.

\bibitem[Dickerson et~al.(2019)Dickerson, Sankararaman, Srinivasan, and
  Xu]{dickerson2019balancing}
John~P Dickerson, Karthik~Abinav Sankararaman, Aravind Srinivasan, and Pan Xu.
\newblock Balancing relevance and diversity in online bipartite matching via
  submodularity.
\newblock In \emph{Proceedings of the AAAI Conference on Artificial
  Intelligence}, volume~33, pages 1877--1884, 2019.

\bibitem[Esfandiari et~al.(2016)Esfandiari, Korula, and
  Mirrokni]{nips16Hossein}
Hossein Esfandiari, Nitish Korula, and Vahab Mirrokni.
\newblock Bi-objective online matching and submodular allocations.
\newblock In \emph{Advances in Neural Information Processing Systems}, pages
  2739--2747, 2016.

\bibitem[Korula et~al.(2018)Korula, Mirrokni, and
  Zadimoghaddam]{korula2018online}
Nitish Korula, Vahab Mirrokni, and Morteza Zadimoghaddam.
\newblock Online submodular welfare maximization: Greedy beats 1/2 in random
  order.
\newblock \emph{SIAM Journal on Computing}, 47\penalty0 (3):\penalty0
  1056--1086, 2018.

\bibitem[Rawitz and Ros{\'e}n(2016)]{rawitz2016online}
Dror Rawitz and Adi Ros{\'e}n.
\newblock Online budgeted maximum coverage.
\newblock In \emph{24th Annual European Symposium on Algorithms (ESA 2016)}.
  Schloss Dagstuhl-Leibniz-Zentrum fuer Informatik, 2016.

\bibitem[Anari et~al.(2017)Anari, Haghtalab, Naor, Pokutta, Singh, and
  Torrico]{anari2017robust}
Nima Anari, Nika Haghtalab, Joseph Naor, Sebastian Pokutta, Mohit Singh, and
  Alfredo Torrico.
\newblock Robust submodular maximization: Offline and online algorithms.
\newblock \emph{arXiv preprint arXiv:1710.04740}, 2017.

\bibitem[Buchbinder et~al.(2015)Buchbinder, Feldman, and
  Schwartz]{buchbinder-15}
Niv Buchbinder, Moran Feldman, and Roy Schwartz.
\newblock Online submodular maximization with preemption.
\newblock In \emph{SODA}, 2015.

\bibitem[Chan et~al.(2017)Chan, Huang, Jiang, Kang, and Tang]{chan2017}
TH~Chan, Zhiyi Huang, Shaofeng H-C Jiang, Ning Kang, and Zhihao~Gavin Tang.
\newblock Online submodular maximization with free disposal: Randomization
  beats 1/4 for partition matroids.
\newblock In \emph{Proceedings of the Twenty-Eighth Annual ACM-SIAM Symposium
  on Discrete Algorithms}, pages 1204--1223. Society for Industrial and Applied
  Mathematics, 2017.

\bibitem[Canonne(2020)]{pois-tail}
Cl\'ement Canonne.
\newblock A short note on poisson tail bounds.
\newblock
  \url{http://www.cs.columbia.edu/~ccanonne/files/misc/2017-poissonconcentration.pdf},
  2020.
\newblock Accessed: 2020-02-01.

\bibitem[Chekuri et~al.(2007)Chekuri, Calinescu, Pál, and Vondrák]{IPCO}
Chandra Chekuri, Gruia Calinescu, Martin Pál, and Jan Vondrák.
\newblock Maximizing a submodular set function subject to a matroid constraint.
\newblock In \emph{Proceedings of the Twelfth Conference on Integer Programming
  and Combinatorial Optimization (IPCO) 2007}, 2007.
\newblock URL \url{http://www.cs.uiuc.edu/homes/chekuri/papers/submod_max.pdf}.

\bibitem[Karimzadehgan and Zhai(2009)]{karimzadehgan2009constrained}
Maryam Karimzadehgan and ChengXiang Zhai.
\newblock Constrained multi-aspect expertise matching for committee review
  assignment.
\newblock In \emph{Proceedings of the 18th ACM conference on Information and
  knowledge management}, pages 1697--1700, 2009.

\bibitem[Goldberg and Tarjan(1988)]{goldberg1988new}
Andrew~V Goldberg and Robert~E Tarjan.
\newblock A new approach to the maximum-flow problem.
\newblock \emph{Journal of the ACM (JACM)}, 35\penalty0 (4):\penalty0 921--940,
  1988.

\bibitem[Nemhauser et~al.(1978)Nemhauser, Wolsey, and
  Fisher]{nemhauser1978analysis}
George~L Nemhauser, Laurence~A Wolsey, and Marshall~L Fisher.
\newblock An analysis of approximations for maximizing submodular set
  functions—i.
\newblock \emph{Mathematical programming}, 14\penalty0 (1):\penalty0 265--294,
  1978.

\bibitem[Harper and Konstan(2015)]{harper2015movielens}
F~Maxwell Harper and Joseph~A Konstan.
\newblock The movielens datasets: History and context.
\newblock \emph{Acm transactions on interactive intelligent systems (tiis)},
  5\penalty0 (4):\penalty0 1--19, 2015.

\bibitem[Cohen et~al.(2019)Cohen, Lee, and Song]{Cohen2019SolvingLP}
Michael~B. Cohen, Y.~Lee, and Zhao Song.
\newblock Solving linear programs in the current matrix multiplication time.
\newblock \emph{Proceedings of the 51st Annual ACM SIGACT Symposium on Theory
  of Computing}, 2019.

\end{thebibliography}
\end{small}

\appendix
\onecolumn
{\Large \redd{ \textbf{Appendix}}}
\vspace{0.1in}
\section{Further Comments on the Differences of Models Proposed in This Paper from Existing Ones}\label{sec:app-a}

\subsection{Comparing \off and Submodular Welfare Maximization (SWM) in~\citet{vondrak2008optimal}}
It is tempting to cast \off as a special case of SWM as follows. First, create $b_j$ copies for each worker $j \in J$. For \bluee{ease} of notation, we still use $J$ to denote the resulting set of workers that may include copies of workers.  Second, for each task $i$, define an uncapacitated utility function as follows:
\begin{align}\label{eqn:app-a}
\tilde{g}_i(S)=\max_{S': S' \subseteq S \cap \cN_i, |S'| \le b_i} g_i(S'), ~~\forall S \subseteq J.
\end{align}
In this way, the original problem of \off can be restated as to find a partition $\{S_i\}$ of $J$ such that $\sum_i \tilde{g}_i(S_i)$ is maximized, where each $\tilde{g}_i$ is an uncapacitated monotone submodular function over $J$. 

The above reduction suffers the issue that we can only get a $(1-1/\sfe)^2$-approximate algorithm. Suppose by applying the classical algorithm in~\citet{vondrak2008optimal}, we get a partition of $J$, say $\cS=\{S_i\}$, such that
\[
\sum_{ i \in I}  \tilde{g}_i(S_i) \ge (1-1/\sfe) \cdot \OPT,
\]
where $\OPT$ denotes the utility by an optimal. Note that $\cS=\{S_i\}$ is not directly feasible to \off: We need to solve another knapsack-constrained monotone-submodular-maximization problem to retrieve $S'$ from $S_i$ for each $i \in I$ according to the definition of $\tilde{g}_i$~\eqref{eqn:app-a}, which incurs an extra approximate factor of $1-1/\sfe$. Thus, we claim that the reduction above yields a $(1-1/\sfe)^2$-approximate algorithm, which is much worse than what is presented in the paper.

\subsection{Comparing \on and Online Submodular Welfare Maximization by~\citet{kapralov2013online}}

The same reduction, as shown above, by introducing an uncapacitated utility function for each task $i \in I$, as defined in~\eqref{eqn:app-a}, is not applicable in the online setting. Suppose we apply~\gre, as proposed by~\citet{kapralov2013online}, to the reduced online SWM, and let $\cS$ be a random set assigned to task $i$ by \gre with $\tilde{g}_i$ defined in~\eqref{eqn:app-a}. It is important to note that the model \on proposed here assumes an online setting \emph{without free disposal}. This implies that for any $\cS=S$, we can only retrieve the first $b_i$ items from $S$. In other words, we do not have access to any items beyond that limit due to the matching capacity of $b_i$.

   \section{Proof of Theorem~\ref{thm:alg2}}
 \label{sec:alg2}
 Here are several  lemmas we need in the proof of  Theorem~\ref{thm:alg2}. WLOG assume $\cN_i=J=[n]\doteq\{1,2,\ldots,n\}$. 
For any vector $\x=(x_j) \in [0,1]^n$ with $\sum_{j \in J} x_j=1$, let  $\pi(\x)  \subseteq J$ be a random set formed by sampling \emph{one single} element from $J$ following the distribution $\x$, while  $\sig(\x)  \subseteq J$ be a random set formed by \emph{independently} sampling each element $j \in J$ with probability $x_j$. Note that $\pi(\x)$ with probability one will have one single element, while $\sig(\x)$ might have multiple elements, though the two share the same marginal distribution. Let $\{\pi^{(\iota)}(\x) | 1\le \iota \le \ell\}$ and $\{\sig^{(\iota)}(\x) | 1\le \iota \le \ell\}$ be $\ell$ \iid copies of $\pi(\x)$ and $\sig(\x)$, respectively. Consider a given non-negative monotone submodular set function $g$ over $J$ and a given integer $\ell$. A key lemma is stated as follows.

\begin{lemma}\label{lem:alg2-a}
\begin{align}
\E\bb{g\bp{\bigcup_{\iota=1}^\ell \pi^{(\iota)}(\x)}} &\ge \E\bb{g\bp{\bigcup_{\iota=1}^\ell \sig^{(\iota)}(\x)}}\label{ineq:alg2-a}.
\end{align}
\end{lemma}
The proof of the above lemma mainly exploits the idea of \emph{swap rounding} as introduced in \cite{chekuri2010dependent}.

\begin{proof}
 For each item $j \in J$, we create $\ell$ copies and let the final set be $J'$ which is a multiset. Thus, each $S \subseteq J'$ can be viewed as a multiset of $J$ with each item having at most $\ell$ copies. For any $S \subseteq J'$, define $g'(S)=g(\sigma(S))$, where $\sig(S)$ is the set of distinct items of $J$ included in $S$. We can verify that $g'$ is submodular function over the ground set $J'$. Observe that $J'$ can be viewed as a union of $\ell$ copies of the set $J$, say $J^{1}, \ldots, J^\ell$. Consider such as a sampling process as follows: For each round $\iota \in [\ell]\doteq\{1,2,\ldots,\ell\}$, we sample an item $j$ from $J^{\iota}$ with probability $x_j$ and let $\cS$ be the random set of all items sampled (copies counted) and $\Y \in \{0,1\}^{n\ell}$ be the characteristic vector of $\cS$. Let $\y \in [0,1]^{n \ell}$ be the concatenation of $\ell$ identical copies of $\x$. Observe that (1) $\E[g'(\Y)]=\E[g'(\cS)]=\E[g(\cup_\iota \pi^{(\iota)}(\x))]$, according to the definition of $g'$ and $\pi^{(\iota)}(\x)$. (2) $\E[g'(\Y)] \ge \sG'(\y)$, where $\sG'$ is the multilinear relaxation of $g'$. Here we try to interpret $\Y$ alternatively as follows.  During each round $\iota \in [\ell]$, we sample an item $j$ from $J^\iota$ with probability $x_j$. This sampling process can be re-interpreted as conducting a series of randomized swap rounding to $\y$ as shown in Lemma VI.2. on page 9 of~\cite{chekuri2010dependent}.  (3) $G'(\y)=\E[g'(\widehat{\cS})]=\E[g(\cup_{\iota} \sig^{(\iota)}(\x))]$, where $\widehat{\cS}$ refers to the random set by independently sampling each element in $J'$ with probability $x_j$. Thus, we are done.
 \end{proof}

The second key lemma involves concepts of multilinear extension and concave closure introduced by~\cite{IPCO}. For any given non-negative monotone submodular set function $g$ over $J$, its multilinear extension, denoted by $\sG$, is defined as $\sG(\x)=\E[g(\sig(\x))]$ for any $\x \in [0,1]^n$. Recall that $\sig(\x)\subseteq J$ denotes a random set formed by independently sampling each element $j \in J$ with probability $x_j$ (note that $\sum_i x_i$ may not necessarily be $1$ here). The concave closure of $g$, denoted by $g^{+}$, is defined as $g^{+}(\x)=\max_{\bfD \in \cD(\x)} \E_{S \sim \bfD} [g(S)]$, where $\cD(\x)$ refers to the collection of all possible distributions over $2^J$ with the marginal distribution equal to $\x$.

\begin{lemma}[Lemma 4.9 on page 23 of~\cite{CRS-14}]\label{lem:alg2-b}
\begin{align}
\sG(b \cdot \x) \ge \big(1-\sfe^{-b}\big) g^{+}(\x), \forall \x \in [0,1]^n, b \in [0,1]
 \label{ineq:alg2-b}.
\end{align}
\end{lemma}

The third lemma involves another extension of $g$, denoted by $g^*$, as introduced in~\cite{IPCO}. It is defined as $g^*(\x)=\min_{S \subseteq J } \bp{g(S)+\sum_{j \in J} x_j \cdot g_S(j)}$ for any $\x \in [0,1]^n$, where $g_S(j)=g(S \cup\{j\})-g(S)$. \cite{IPCO} shows another key lemma below.

\begin{lemma}[Lemma 4 on page 9 of~\cite{IPCO}]\label{lem:alg2-c}
\begin{align}
g^*(\x) \ge g^{+}(\x), \forall \x \in [0,1]^n  \label{ineq:alg2-c}.
\end{align}
\end{lemma}

Now we start to prove Theorem~\ref{thm:alg2}.
\begin{proof}
Assume that the total number of arrivals $A=\ell$ in \bbsp. Let $\cS_\ell$ denote the random set of distinct balls which arrived before termination in \bbsp. Note that during any round $t$, a ball $j$ will arrive with probability $(y_{ij}^*/T)/(\sum_{e'\in E_i} y_{e'}^*/T)=y_{ij}^*/b$, conditioning on one arrival during $t$.  Right now we have $\ell$ arrivals in total, the arrival distributions are independent over all those $\ell$ rounds. Thus, we claim that the set of distinct arrival balls $\cS_\ell$ should be exactly equal to $\bigcup_{\iota=1}^\ell \pi^{(\iota)}(\y^*/b)$, where $\y^*=\{y_{ij}^*\}$. From Lemma~\ref{lem:alg2-a}, we see that $\E[g_i(\cS_\ell) ]\ge \E\bb{g_i\bp{\bigcup_{\iota=1}^\ell \sig^{(\iota)}(\y^*/b)}}$.

Let $\Sig(\y^*/b)=\bigcup_{\iota=1}^\ell \sig^{(\iota)}(\y^*/b)$, which denotes the $\ell$ \iid copies of  $\sig^{(\iota)}(\y/b)$. Note that for each $j \in J$,
\begin{align}
\Pr[j \in \Sig(\y^*/b)]=1-\big(1-\frac{y_{ij}^*}{b}\big)^\ell \ge  \bp{1-\big(1-\frac{1}{b}\big)^\ell}y_{ij}^* \doteq \kap(\ell) \cdot y_{ij}^*. \label{ineq:alg2-1}
\end{align}
 The last inequality above is due to the fact that $h(x) \doteq  1-\big(1-\frac{x}{b}\big)^\ell$ is an increasing and concave function over $x\in [0,1]$ with $h(0)=0$. Thus, $h(x) \ge h(1) \cdot x$ for all $x \in [0,1]$. Inequality~\eqref{ineq:alg2-1} suggests that $\Sig(\y^*/b)$ includes each element $j \in J$ independently with a marginal probability at least $\kap(\ell) \cdot y_{ij}^*$. Recall that $\sG_i$ is the multilinear relaxation of $g_i$. 
  Therefore, we have 
 \begin{align}
 \E[g_i(\cS_\ell) ] & \ge \E\bb{g_i\bp{\bigcup_{\iota=1}^\ell \sig^{(\iota)}(\y^*/b)}} =\E\bb{g_i\bp{\Sig(\y^*/b) }} \\
 &\ge \sG_i\bp{\kap(\ell) \cdot \y^*} \ge \big(1-\sfe^{-\kap(\ell)}\big) g_i^{+}(\y^*) 
 \ge  \big(1-\sfe^{-\kap(\ell)}\big) \OPT_i.
\label{ineq:alg2-2}
\end{align}
 The first inequality in~\eqref{ineq:alg2-2} is due to the monotonicity of $g_i$. The second inequality in~\eqref{ineq:alg2-2} follows from Lemma~\ref{lem:alg2-b}. Note that the distribution selected by the clairvoyant optimal over $2^J$ is $\{x^*_{i,S}\}$,  which has the marginal distribution equal to $\{y_{ij}^*\}$ according to \nna. Thus, by definition of the concave closure, we  get the third inequality in~\eqref{ineq:alg2-2}. Therefore, we establish the result of Theorem~\ref{thm:alg2} for the general case $2 \le \ell \le b$.
 
Note that for $A=\ell=1$, we have that 
\[
\E[g_i(\cS_1)]=\sum_{j =1}^n \frac{y_{ij}^*}{b} \cdot g_i(\{j\}) \ge  \frac{g_i^*(\y^*)}{b}
\ge  \frac{g_i^{+}(\y^*)}{b} \ge \frac{\OPT_i}{b}.
\]
The first inequality above follows from the definition of $g_i^*$ by setting $S=\emptyset$ (note that $g_i(\emptyset)=0$). The second inequality above is due to Lemma~\ref{lem:alg2-c}, while the last one  follows from the definition of $g^{+}$. 
\end{proof}

\end{document}